\documentclass[twocolumn,10pt]{IEEEtran}
\usepackage{amsmath}
\usepackage{amssymb}
\usepackage{amsfonts}
\usepackage{epsfig}
\usepackage{cite,soul}

\def\bfh{{\boldsymbol{h}}}

\def\bfr{{\boldsymbol{r}}}

\def\bfs{{\boldsymbol{s}}}

\def\bfw{{\boldsymbol{w}}}

\def\bfx{{\boldsymbol{x}}}

\def\bfy{{\boldsymbol{y}}}

\def\bfG{{\boldsymbol{G}}}

\def\bfH{{\boldsymbol{H}}}

\def\bfM{{\boldsymbol{M}}}

\def\bfQ{{\boldsymbol{Q}}}

\def\bfR{{\boldsymbol{R}}}

\def\bfW{{\boldsymbol{W}}}

\def\bfX{{\boldsymbol{X}}}

\def\bfY{{\boldsymbol{Y}}}

\def\bfalpha{{\boldsymbol{\alpha}}}

\def\bfPi{{\boldsymbol{\Pi}}}

\newcommand{\bit}{\begin{itemize}}
\newcommand{\eit}{\end{itemize}}

\newcommand{\bc}{\begin{center}}
\newcommand{\ec}{\end{center}}

\newcommand{\ba}{\begin{array}}
\newcommand{\ea}{\end{array}}

\newcommand{\beq}{\begin{equation}}
\newcommand{\eeq}{\end{equation}}

\newcommand{\beqn}{\begin{equation*}}
\newcommand{\eeqn}{\end{equation*}}

\newcommand{\bean}{\begin{eqnarray*}}
\newcommand{\eean}{\end{eqnarray*}}
\newcommand{\bea}{\begin{eqnarray}}
\newcommand{\eea}{\end{eqnarray}}

\def\Z{\mathbb{Z}}

\def\R{\mathbb{R}}
\def\C{\mathbb{C}}
\def\E{\mathbb{E}}

\def\Im{\boldsymbol{I}}

\DeclareGraphicsExtensions{.pdf}%

\newtheorem{thm}{Theorem}

\newtheorem{cor}[thm]{Corollary}
\newtheorem{note}{Remark}

\newcommand{\mac}{{\tm{mac}}}
\newcommand{\ext}{{\tm{ext}}}
\def\snr{\textnormal{SNR}}
\newcommand{\tm}[1]{\textnormal{#1}}
\def\im{\, \imath \,}
\def\rank{\textnormal{rank}}
\def\trace{\textnormal{Tr}}
\def\CN#1#2{\C {\cal N} \left( #1, #2 \right)}
\def\R{{\mathbb R}}

\def\Z{{\mathbb Z}}

\def\E{{\mathbb E}}
\def\C{{\mathbb C}}
\renewcommand{\Re}{\textnormal{Re}}
\renewcommand{\Im}{\textnormal{Im}}

\def\norm#1{\left\| #1 \right\|}
\def\abs#1{\left| #1 \right|}
\newenvironment{proof}{\begin{IEEEproof}}{\end{IEEEproof}}
\newcommand{\qed}{\hfill\IEEEQED}

\title{Performance-Complexity Analysis for MAC ML-based Decoding with User Selection}

\author{Hsiao-feng (Francis) Lu, Petros Elia, and Arun Singh
\thanks{This paper was presented in part at IEEE VTS APWCS 2014, Ping Tung, Taiwan, August 28-29, 2014.}
}
\usepackage{color}

\begin{document}
\bibliographystyle{ieeetran}
\allowdisplaybreaks
\maketitle

\begin{abstract}
This work explores the rate-reliability-complexity limits of the quasi-static $K$-user \emph{multiple access channel} (MAC), with or without feedback. Using high-SNR asymptotics, the work first derives bounds on the computational resources required to achieve near-optimal (ML-based) decoding performance.  It then bounds the (reduced) complexity needed to achieve any (including suboptimal) diversity-multiplexing performance tradeoff (DMT) performance, and finally bounds the same complexity, in the presence of feedback-aided user selection. This latter effort reveals the ability of a few bits of feedback not only to improve performance, but also to reduce complexity. In this context, the analysis reveals the interesting finding that proper calibration of user selection can allow for near-optimal ML-based decoding, with complexity that need not scale exponentially in the total number of codeword bits. The derived bounds constitute the best known performance-vs-complexity behavior to date for ML-based MAC decoding, as well as a first exploration of the complexity-feedback-performance interdependencies in multiuser settings.
\end{abstract}

\section{Introduction}

\subsection{Multiple access system model} \label{sec:sys_model}

We consider a symmetric multiple access channel (MAC) with $K$ single-antenna users, communicating to a receiver with $n_r$ receiving antennas, over a quasi-static fading channel. Each user $i$, $i=1,2,\ldots,K$ communicates over the same duration of $T$ time slots, while the receiver accumulates an $n_r \times T$ signal matrix $\bfY$ taking the form
\begin{equation}
\bfY=\sqrt{\snr} \sum_{i=1}^K \bfh_i \, \bfx_i^\top +\bfW \label{eq:sys1_mac}
\end{equation}
where $\bfh_i \sim \CN{\mathbf{0}}{\boldsymbol{I}_{n_r}}$ is the $i$th user channel vector with independent identically distributed (i.i.d.) zero-mean unit-variance $\CN{0}{1}$ complex Gaussian entries, where $\bfW$ is the received noise matrix with i.i.d. $\CN{0}{1}$ entries, where $\snr$ denotes the signal to noise ratio, and where ${\bfx_i}=[x_{i,1} \cdots x_{i,T}]^\top$ is a scaled version of a $T$-length code vector of user $i$ satisfying an average power constraint $\E \norm{\bfx_i}^2 \leq T, \ i=1, \ldots, K$.

\subsubsection{Exploring the scenario of outage-limited communications with bounded computational resources}
We here consider the outage-limited MAC setting, where the channel in \eqref{eq:sys1_mac} is randomly drawn but it remains fixed throughout the coding duration of $T$ channel uses. This common assumption corresponds to scenarios where the channel changes slowly, and where communication takes place under strict latency constraints that do not allow for encoding over a large number of fading realizations.

This same outage-limited setting can often experience reduced reliability due to the event of outage where the instantaneous channel cannot support the user rates, and due to the fact that in cases like the MAC, the maximum allowable rates are diminished by interference. Moreover, to meet the strong demands for faster data rates on wireless and cellular channels, most communications systems would opt to operate closer and closer to the theoretical limits of capacity, at the further expense of reliability. This in turn, naturally forces the need for high-performance transceivers that will not further sacrifice  reliability performance, but which can often be computationally expensive, or even computationally prohibitive.

What happens though to this performance if communication takes place under strict constraints on the computational resources that can be used to encode and decode? Equivalently, one can ask, what computational resources are needed to achieve a certain rate-reliability performance. The main objective of our work is to provide some understanding of such complexity-performance interdependencies that are crucial in the MAC.

\subsubsection{Complexity-performance interdependencies in the MAC.}
Naturally there are many such interdependencies between the key parameters of SNR, rate, computational complexity and reliability. For example it is easy to imagine that, typically, decreasing computational resources can potentially increase the probability of error by placing a constraint on the coding duration $T$, by forcing the use of less complex transceivers with suboptimal performance, or even by requiring that the decoding effort be terminated early (computational halt) at the risk of additional errors. Similarly, increased user rates can mean larger and denser codes, with more errors and larger decoding algorithmic efforts. Along the same lines, fixing the rate and increasing the SNR, will make the codewords sparser and thus possibly easier to differentiate and decode.

\subsection{Complexity-performance measures and high-SNR approximations}
Let $N_{\max}$ denote the amount of computational reserves, in floating point operations (flops) per $T$ channel uses, that the system is endowed with, in the sense that after $N_{\max}$ flops, the decoder must simply terminate, potentially prematurely and before completion of the task, thus declaring an error. Also let $P_e$ be the probability of error associated to the decoder, in the presence of the aforementioned computational constraints.
Motivated by the need for reduced-complexity high-performance ML-based decoders (cf.~\cite{JitRaj13,NatSriRaj13,JalBarOtt09,HowSirCal,NatRaj13b}), we explore the properties of this quadruple $\bigl(\snr, R,P_e,N_{\max}\bigr)$, for ML-based MAC decoding.
\subsubsection{High SNR analysis of performance and complexity}
Towards making sense of the complexity-vs-performance interdependencies, we will exploit carefully-chosen asymptotic bounds, where analytical tools can rigorously and concisely approximate these interdependencies, allowing for cleaner insights that can yield impact. Such asymptotic bounds (generally mapping a random problem to its near-deterministic limits), help to identify where the largest gaps in our understanding may lie. We will here focus on \emph{high-SNR} asymptotic bounds but note that the derived approximations can offer insights in operational regimes of moderate-to-high SNR values.
Such emphasis on the moderate-to-high SNR values, can better capture the core of the complexity-vs-performance problem because such operational regions can eventually support higher rates, corresponding to larger codebooks that are typically --- but not always, as we will surprisingly see later on --- harder to decode.

Towards this, we let each user's rate scale as $R=r \log_2 \snr$, where $r$ denotes the \emph{multiplexing gain} \cite{ZheTse} describing the density of the constellation, and serving as a measure of how far each user's rate is from the capacity of an AWGN channel. Each user $i$ employs a code ${\cal X}_{r,i}$ of cardinality $|{\cal X}_{r,i}| = 2^{RT} = \snr^{rT}$, and the receiver uses a decoder $\mathcal{D}_r$, which is considered to make an error if any of the $K$ users' messages is not decoded correctly. We restrict our attention to the class of lattice codes and joint ML decoders, which we describe in Section~\ref{sec:latticeEncoderDecoder} that follows.

\subsubsection{Complexity and reliability exponents}
For complexity analysis we adopt the high-SNR approach in \cite{SinghEJ12,JaldenE12}, where for some $r$, for some encoders $\{\mathcal{X}_{r,i}\}_{i=1}^K$, and for a decoder $\mathcal{D}_r$, the \emph{complexity exponent} $c(r)$ was defined to be
\begin{equation}
c(r) : =   \lim_{\snr \to \infty} \frac{\log N_{\max}(r)}{\log \snr} \label{eq:complexityExponent}.
\end{equation}
We also adopt the well-known DMT approach of \cite{ZheTse} where, in the same high SNR regime, the \emph{diversity gain} $d(r)$ takes the form
\begin{equation}
d(r):= - \lim_{\snr \to \infty} \frac{\log P_e(r)}{\log \snr}. \label{eq:DMT1}
\end{equation}
Using this asymptotic approach, the work in \cite{TseVisZhe} has shown that the optimal $K$-user MAC DMT performance, takes the form
\begin{equation}
d^*_{\tm{mac}}(r) := \begin{cases}
            n_r(1-r), &  \tm{ if } 0 \leq r \leq  \frac{n_r}{K+1},\\
            d_{K,n_r}^*(Kr), & \tm{ if } \frac{n_r}{K+1} < r \leq \frac{n_r}{K},
\end{cases}
\label{eq:dmt_mac_user}
\end{equation}
where $d_{m,n}^*(r)$ denotes the optimal DMT of an $(m \times n)$ MIMO channel (cf.\cite{ZheTse} for its exact characterization).
Our first effort here will be to bound the complexity exponent $c(r)$ that can guarantee any MAC-DMT  performance $d(r)\leq d^*_{\tm{mac}}(r)$.
Before we do that, let us try to get some insight on these two competing exponents.

\subsubsection{Insight on $d(r)$ and $c(r)$}
To gain some insight, we note that in terms of the MAC-DMT performance, the above expression in~\eqref{eq:dmt_mac_user} reflects the existence of two distinct $r$ regions; the \emph{lightly-loaded} multiplexing gain region $0 \leq r \le  \frac{n_r}{K+1}$ where the MAC exhibits single-user behavior as if there were no multiuser interference, and the \emph{heavily-loaded} region $\frac{n_r}{K+1} < r \leq \frac{n_r}{K}$, where the MAC exhibits an antenna-pooling behavior.

In terms of the complexity exponent $c(r)$, it is easy to see that in this MAC setting, there is a total of $2^{RKT} = \snr^{rKT}$ codewords (jointly from all users), which means that a \emph{brute-force} joint ML decoder would always make $2^{RKT}$ codeword visits~\footnote{The number of codeword visits is ---  in the high SNR setting --- in the same order as the number of computational flops.}. Thus in terms of exponents, this means that such a brute-force optimal ML decoder, can achieve the optimal $d^*_{\tm{mac}}(r)$ with a required complexity exponent $c(r) = rKT$ that is indeed the maximum (meaningful) complexity exponent in this setting\footnote{Considering decoders with higher complexity than brute-force ML, is unnecessary because ML decoders are already optimal.}. We will show that a properly designed sphere decoder can achieve this same ML performance, with a reduced complexity exponent $c(r) < rKT$. To give the reader an idea, we can recall from the work in \cite{JaldenE12} that, if for example, there is only one user  ($K=1$) who has $n_t$ transmit antennas, and the receiver has $n_r\geq n_t$ antennas, then  --- again in the high SNR limit ---  this same ML performance can be achieved with computational resources that are substantially smaller than brute force, and which corresponded to an (effective) complexity exponent that was a non-monotonic\footnote{This non monotonic expression is maximum somewhere in the mid-ranges of $r$, suggesting that the computational resources to achieve the asymptotically optimal complexity need not necessarily increase with increasing rate. For example, for the simple case of $n_t = 2$, in the presence of the minimum required $T= n_t=2$ in order to achieve the optimal DMT of $d_{2,n_r}^*(r)$, the complexity exponent was increasing as $c(r) = r$ for $0\leq r\leq 1$, and was then decreasing as $c(r) = 2-r$ for $1\leq r \leq 2$.}  piece-wise linear function in $r$, of the form
\[
c(r) = r T \left(1- \frac{r}{n_t} \right), \ \ r =0,1,2,\ldots,n_t.
\]
We will do something similar here, for the multiuser case $K\geq 1$, and we will also explore the effect of a few bits of feedback on the $d(r)$ and $c(r)$.

\subsection{Notation and assumptions}
Following \cite{ZheTse}, we use $\doteq$ to denote the \emph{exponential equality}, i.e., a function $f(\snr)$ is said to be $f(\snr)\doteq \snr^{b}$ if and only if $\lim_{\snr \to \infty} \frac{\log f(\snr)}{\log \snr}=b$. Exponential inequalities such as $ \dot\leq $, $ \dot\geq $ are similarly defined. By $s=\lceil x\rceil$ we mean the smallest integer $s \geq x$, and by $t=\lfloor x\rfloor$ we mean the largest integer $t \leq x$. Capital boldface letters are reserved for matrices, and the lower-case boldface ones are for column vectors. ${\bf A}^\dag$ is the Hermitian transpose of matrix $\bf A$, and  $(x)^+ := \max\{x,0\}$. Finally, we define $\nu:=\min\{K,n_r\}$.

In terms of assumptions, we consider fading coefficients that have a circularly symmetric complex-normal distribution and which are i.i.d. in space. As stated, the coding duration is smaller than the coherence interval of the channel, and hence the fading is randomly drawn, but it is held fixed throughout the communication process. We will assume that the receiver has full channel state information, i.e., knows completely the channel vectors $\bfh_i$ of every user $i$, while the users have no knowledge of $\bfh_i$ in the case of no feedback. Furthermore, the rate of communication $R$ is kept constant and thus does not change as a function of the channel and of the feedback. Finally, different assumptions regarding the structure of the code and decoder, will be presented immediately below in Section~\ref{sec:latticeEncoderDecoder}. Assumptions on the feedback-aided user-selection algorithm will be presented in Sec. \ref{sec:FB1}.

\subsection{Lattice encoders and joint decoder\label{sec:latticeEncoderDecoder}}

Directly from \eqref{eq:sys1_mac} we have the joint real-valued vectorized model which takes the form
\begin{equation}
\bfy= \sqrt{\snr} \bfH \bfx + \bfw
\label{eq:vect}
\end{equation}
where
\begin{equation}
\bfH = \boldsymbol{I}_T \otimes \left[
\begin{array}{rr}
\Re\left\{ \bfH_{\tm{\!eq}}\right\}& -\Im\left\{ \bfH_{\tm{\!eq}}\right\}\\
\Im\left\{ \bfH_{\tm{\!eq}}\right\} & \Re\left\{ \bfH_{\tm{\!eq}}\right\}
\end{array}\right]  \label{eq:H}
\end{equation}
where
\[
\bfH_{\tm{\!eq}}=[\bfh_1 \ \cdots  \ \bfh_K]
\]
represents the entire $(n_r \times K)$ channel fading matrix, where
\begin{multline*}
\bfx=\left[ \Re\{ x_{1,1}\} \, \Im\{ x_{1,1}\} \cdots \Re\{ x_{K,1}\}\, \Im\{ x_{K,1}\}\right.\\
\left. \,  \Re\{ x_{1,2}\} \, \Im\{ x_{1,2}\}\,  \cdots\,  \Re\{ x_{K,T}\} \, \Im\{ x_{K,T}\} \right]^\top
\end{multline*}
is the joint codeword vector aggregated over all users, and where $\bfy$ and $\bfw$ are defined similarly to $\bfx$. The joint codeword $\bfx$ is decoded using a joint ML-based decoder which, in its brute-force form, is MAC DMT optimal (cf.~\cite{TseVisZhe, Kuser}), i.e., achieves the optimal MAC DMT $d_\mac^*(r)$.
$\bfx$ is taken from a (sequence of) full-rate linear (lattice) code(s) $\mathcal{X}_r= {\cal X}_{r,1} \oplus \cdots \oplus {\cal X}_{r,K}$, where ${\cal X}_{r,i}=\Lambda_{r,i} \cap {\cal R}_{r,i} \subset \R^{2T}$ is the corresponding lattice code for the $i$th user consisting of those elements of the rank $2T$ lattice $\Lambda_{r,i}$ that lie inside the {\em shaping region} ${\cal R}_{r,i}$, which is properly chosen to meet the rate requirement $\abs{{\cal X}_{r,i}}=2^{R T}$ as well as the average power constraint. The region ${\cal R}_{r,i}$ is a compact convex subset of $\R^{2T}$. Specifically, we set $\Lambda_{r,i} := \snr^{-\frac r2}\Lambda_i$, to be a scaled lattice of another lattice $\Lambda_i$ whose generator matrix is denoted by $\bfG_i$. Thus for $\bfG=\tm{diag}(\bfG_1, \ldots, \bfG_K)$, the overall codeword is given by $\bfx \ = \ \snr^{-\frac r2} \bfG \bfs$ for some $\bfs \in \Z^{2KT}$.  Substituting this into~\eqref{eq:vect} yields the following equivalent channel input-output relation which will be used for sphere decoding of $\bfs$
\begin{equation}
\bfy = \bfM \bfs + \bfw \label{eq:newvect_mac}
\end{equation}
where $\bfM := \snr^{\frac {1-r}{2}} \bfH \bfG$.

This joint ML (or lattice) decoder is implemented as a bounded-search sphere decoder (SD). For SD algorithm details, readers are referred to \cite{conf_DaElCa04,SinghEJ12,JaldenE12} however for clarity of exposition, wherever necessary, essential details are provided during the complexity analysis. Implicit to the use of a sphere decoder is a chosen search radius $\delta$, a chosen decoding order corresponding to an order with which symbols of $\bfs$ are decoded, and a time-out policy that terminates the decoder once the algorithm exceeds a certain computational threshold. The termination policies that we use will be clarified depending on the setting. Our results will optimize over $\delta$ by setting $\delta:=\sqrt{z\log \snr}$, for a properly chosen $z>0$.
The idea here is that the search radius should be big enough so that --- loosely speaking --- the probability that AWGN noise has a norm larger than this radius, is sufficiently smaller than the probability of error under brute-force ML. At the same time, this radius needs to remain sufficiently small so that the number of elements within the search sphere is small enough most of the time.
Finally, the bounds will hold irrespective of the decoding order. As a result, we will henceforth limit reference to the search radius and the decoding order, mainly in proofs. As suggested before, the derivations focus on ML-based decoding, but given the aforementioned performance-and-complexity equivalence between ML and \emph{regularized} lattice based decoding \cite{SinghEJ12}, these same results extend automatically to the latter.

We finally note that the validity of the presented bounds depends on the existence of actual coding schemes that meet them, and which will be identified here, together with the associated SD implementation and halting policies. Regarding the codes, we hasten here to say that all rate-reliability results can be achieved with uncoded QAM transmission for any  $n_r$. This is one of the crucial contributions of this paper, and it is presented in Theorem~\ref{thm:qamau}.

\section{Performance-Complexity Tradeoff for MAC ML Decoding \label{sec:ComplPerfTradeoff}}

We first provide an upper bound on the complexity exponent that guarantees a certain MAC-DMT performance $d(r)$.

\vspace{3pt}
\begin{thm}
\label{thm:mac_cr_for_optimalDMT}
For the $K$-user MAC, the minimum over all ML-based decoders (all SD implementations, all halting and all decoding order policies) complexity exponent $c_{\tm{mac},d}(r)$ required to achieve a certain DMT $d(r)\leq d^*_{\tm{mac}}(r)$, is upper bounded by
\begin{equation}
\bar{c}_{\mac,d}(r) =
\left\{
\begin{array}{l}
\sup\limits_{{\boldsymbol \mu} \in {\cal B}(r)} (K-n_r)rT\\
\qquad +T\sum_{i=1}^{\nu}\left(r-(1-\mu_{i})^+\right)^+,\tm{if $K > n_r$,}\\\\
\sup\limits_{\boldsymbol{\mu} \in {\cal B}(r)} T\sum_{i=1}^{\nu}  \left[ \min\left\{ r, r+\mu_i -1\right\}\right]^+,\\
\hspace{1.8in} \tm{if $K \leq n_r$,}
\end{array} \right.
\label{eq:cbar1}
\end{equation}
where $\boldsymbol{\mu}=[\mu_1 \ \cdots \mu_\nu]^\top$, $\nu=\min\{K,n_r\}$ and
\[
{\cal B}(r) := \left\{ \boldsymbol{\mu}:
\begin{array}{l}
\mu_1 \geq \cdots \geq \mu_{\nu}, 0 \leq \mu_i \in \R\\
\sum_{i=1}^{\nu} \left( \abs{K-n_r}+2i-1\right)\mu_i \leq d(r)
\end{array} \right\}.
\]
\end{thm}
\vspace{3pt}
\begin{proof}
See Appendix A.
\end{proof}
\vspace{3pt}

\begin{note}
A remark is in order, regarding the high complexity present in the under-determined MAC case of $n_r < K$, and particularly regarding the term $(K-n_r)rT$ appearing in \eqref{eq:cbar1} (and the same in \eqref{eq:cbar} appearing later for $T=1$), which results in an increased complexity exponent, irrespective of the desired diversity performance. To gain some insight on this, we quickly recall that a sphere decoder performs a QR decomposition of matrix $\bfM$ in~\eqref{eq:newvect_mac}, which though for $n_r < K$, results in a matrix $\bfR$ that is an upper trapezoid matrix, whose bottom row contains $2T(K - n_r)+1$ nonzero entries. Therefore, prior to processing the root node of a sphere-decoding tree, the sphere decoder must first search exhaustively among $N^{2T(K-n_r)}$ combinations of the $N$-ary PAM constellation points of the integer entries of $\bfs$ (cf.~\eqref{eq:newvect_mac}). For our particular case of having single antenna transmitters, $N$ suffices to be chosen as $N=\snr^\frac{r}{2}$, which explains the term $(K-n_r)rT$ in~\eqref{eq:cbar1}.
\end{note}

\subsection{Complexity cost of achieving the optimal MAC-DMT}
We here seek to derive the complexity exponent needed to achieve the MAC-DMT optimal $d^*_{\tm{mac}}(r)$ (cf. \eqref{eq:dmt_mac_user}). To answer this question, we will need to understand the coding duration requirement $T$ for such DMT optimality. Towards this, we recall that a first step towards coding that achieves the optimal DMT performance, was taken in \cite{TseVisZhe}, which showed that for the lightly-loaded multiplexing-gain region $r \le  \frac{n_r}{K+1}$, uncoded (QAM) transmission ($T=1$) is in fact DMT optimal. The question of deriving codes that can achieve DMT optimality for all $r$, was resolved in \cite{Kuser, remarkdmt,LVHLHV}, which provided --- depending on the values of $K$ and $n_r$ --- different DMT optimal lattice coding schemes by encoding over time, thus requiring $T>1$.

The following theorem plays a crucial role in tightening the complexity exponent, by proving that uncoded transmission ($T=1$, QAM) is indeed MAC-DMT optimal for all multiplexing gains. The following result holds for all isotropic channel probability distributions, and it extends the aforementioned result in \cite{Kuser, remarkdmt,LVHLHV,TseVis} to all $r$, all $n_r$ and all $K$.

\vspace{3pt}
\begin{thm} \label{thm:qamau}
Uncoded QAM for each user achieves DMT optimality in the $K$-user symmetric SIMO MAC for all $r,K,n_r$.
\end{thm}
\vspace{3pt}
\begin{proof}
See Appendix \ref{app:qamau}.
\end{proof}
\vspace{3pt}

We can now apply Theorem~\ref{thm:mac_cr_for_optimalDMT} to provide upper bounds on the complexity exponent $c^*_{\tm{mac}}(r)$ that guarantees the DMT optimal $d^*_\mac(r)$. The bound is constructive and is optimized over all known ML-based decoders, and over all existing codes. The result holds for i.i.d. Rayleigh fading statistics and for any decoding order policy. It corresponds to a DMT optimal transceiver that employs QAM transmission with $T=1$, a search radius $\delta > \sqrt{d^*_\mac(r) \log \snr}$, and a decoding halting policy that naturally halts decoding at $\snr^{c^*_\mac(r)}$ flops.

\vspace{3pt}
\begin{thm}
\label{thm:cr_mac}
For the $K$-user MAC, the minimum over all codes and all ML-based decoders (all SD implementations, all halting and all decoding order policies) complexity exponent $c^*_{\tm{mac}}(r)$ required to achieve the optimal DMT $d^*_{\tm{mac}}(r)$, is upper bounded by
\begin{equation}
\bar{c}_\mac(r) = \left\{
\begin{array}{l}
\sup\limits_{\boldsymbol{\mu} \in {\cal B}(r)} (K-n_r)r+\sum_{i=1}^{\nu}\left(r-(1-\mu_{i})^+\right)^+,\\
\hspace{1.8in}\tm{if $K > n_r$}\\\\
\sup\limits_{\boldsymbol{\mu} \in {\cal B}(r)} \sum_{i=1}^{\nu}  \left[ \min\left\{ r, r+\mu_i -1\right\}\right]^+,\\
\hspace{1.8in} \tm{if $K \leq n_r$,}
\end{array} \right.
\label{eq:cbar}
\end{equation}
where $\boldsymbol{\mu}=[\mu_1 \ \cdots \mu_\nu]^\top$, $\nu=\min\{K,n_r\}$  and
\[
{\cal B}(r) := \left\{ \boldsymbol{\mu}:
\begin{array}{l}
\mu_1 \geq \cdots \geq \mu_{\nu}, 0 \leq \mu_i \in \R\\
\sum_{i=1}^{\nu} \left( \abs{K-n_r}+2i-1\right)\mu_i \leq d^*_\mac(r)
\end{array} \right\}.
\]
\end{thm}
\vspace{3pt}
\begin{proof}
It is a direct consequence of Theorems \ref{thm:mac_cr_for_optimalDMT} and \ref{thm:qamau}.
\end{proof}
\vspace{3pt}

The following corollaries provide explicit expressions for $\bar{c}_\mac(r)$  for the single-input single-output (SISO) case of $n_r=1$, and for $n_r=K$.

\begin{cor} \label{cor:nr1}
For the $K$-user SISO MAC ($n_r=1$), the complexity exponent required to achieve the optimal $d^*_{\tm{mac}}(r)$, is upper bounded by
\begin{equation}
\bar{c}_\mac(r) = (K-1)r, \quad \tm{ for } 0 \leq r \leq \frac 1K.
\end{equation}
\end{cor}

\begin{cor} \label{cor:nrK}
For the $K$-user MAC with $n_r=K$, the complexity exponent required to achieve the optimal $d^*_{\tm{mac}}(r)$, is upper bounded by
\begin{multline}
\bar{c}_\mac(r)=r\left\lfloor{\sqrt{K(1-r)}}\right\rfloor\\ +\left(r-1+\frac{K(1-r) -(\left\lfloor{\sqrt{K(1-r)}}\right\rfloor)^2 }{2\left\lfloor{\sqrt{K(1-r)}}\right\rfloor+1}\right)^+.\label{eq:nrK1}
\end{multline}
\end{cor}

Fig.~\ref{fig:mac_nr1} plots $\bar{c}_\mac(r)$ for $K=4,5$ and $n_r=1$, while Fig.~\ref{fig:comp_vblast_nr_3_4} plots $\bar{c}_\mac(r)$ for $n_r=K=3,4,5$. In interpreting the plots below, we recall that the here maximum achievable multiplexing gain is $\frac{n_r}{K}$.

\begin{note}
As we see from \eqref{eq:nrK1} and Fig.~\ref{fig:comp_vblast_nr_3_4}, when $n_r = K$, the complexity exponent $c(r)$ is not monotonically increasing.
This is simply because an increasing $r$ might indeed increase the objective function in the maximization, but at the same time it also decreases $d(r)$ and thus also decreases the volume of ${\cal B}(r)$ over which maximization takes place.
\end{note}

\begin{figure}[ht!]
\[
\includegraphics[width=0.8\columnwidth]{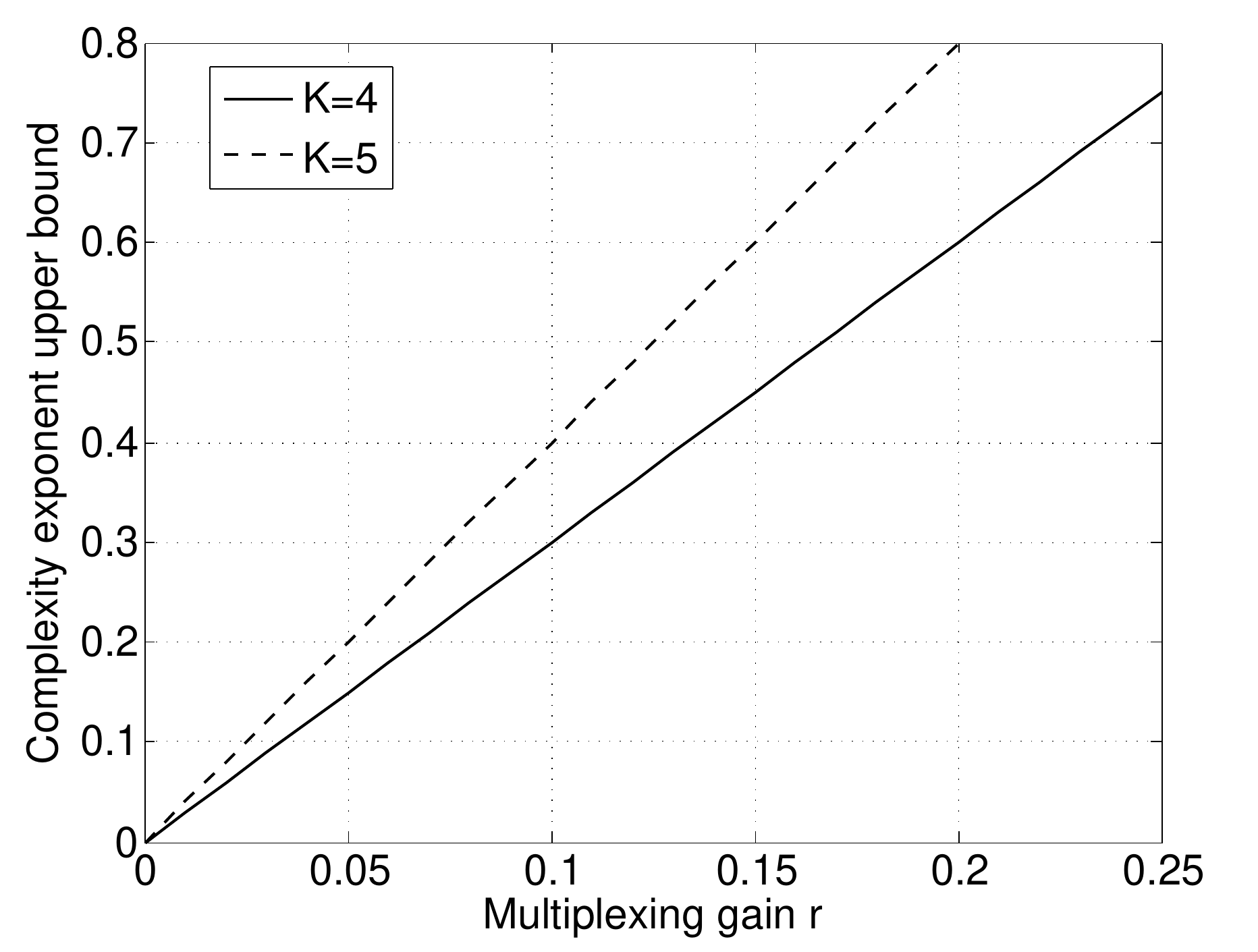}
\]
\caption{Bound on the complexity exponent guaranteeing DMT optimality in the $K$-user SISO MAC: $K=4,5$ and $n_r=1$.}
\label{fig:mac_nr1}
\end{figure}

\begin{figure}[ht!]
\[
\includegraphics[width=0.8\columnwidth]{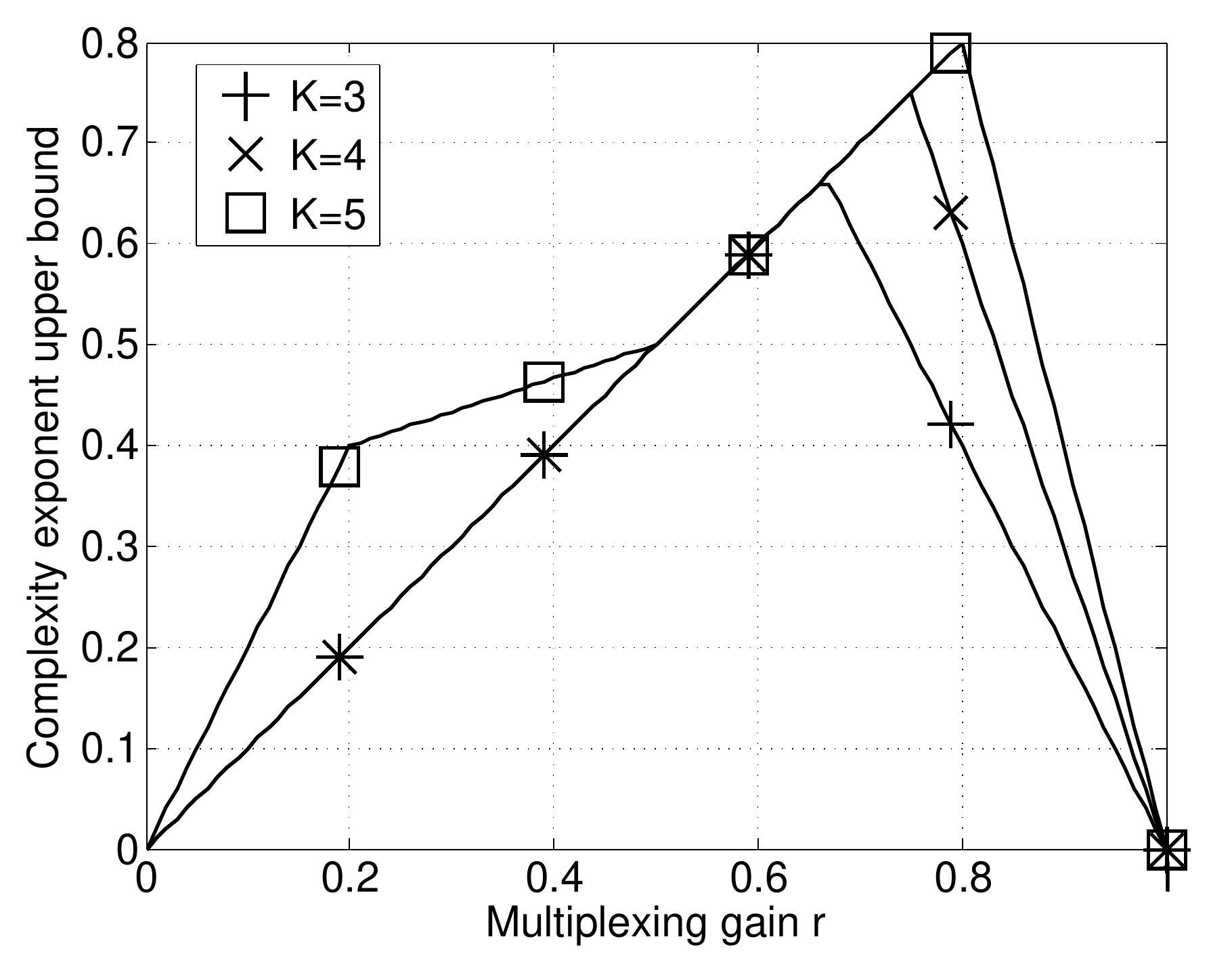}
\]
\caption{Bound on the complexity exponent guaranteeing DMT optimality in the $K$-user SIMO MAC: $K=3,4,5$ and $n_r=K$.}
\label{fig:comp_vblast_nr_3_4}
\end{figure}

\subsection{The benefit of increasing the number of receive antennas}
In the following we show the rather surprising result that a substantial increase in the number $n_r$ of receive antennas, allows for maximal reductions in the complexity exponent. This is indeed surprising because, as $n_r$ increases, a MAC-DMT optimal decoder must handle channels that are increasingly more singular~\footnote{Recall from~\eqref{eq:dmt_mac_user} that for $n_r > K$ the optimal MAC DMT is $d_\mac^*(r)=n_r(1-r)^+$.} and thus typically harder to decode because near-singular channels typically result in denser signaling constellations at the receiver.

\vspace{3pt}

\begin{cor}\label{cor:largenr}
For the $K$-user SIMO MAC with $n_r \gg K$, the complexity exponent for achieving the optimal MAC-DMT $d^*_{\tm{mac}}(r)$, approaches
\[
\bar{c}_\mac(r) =  \min \left\{ r, \frac{K-1}{n_r-K+1}(1-r)\right\}
\]
and in the limit of asymptotically large $n_r$, the complexity exponent tends to zero.
\end{cor}

\vspace{3pt}

\begin{proof}
See Appendix~\ref{app:largenr}.
\end{proof}
\begin{figure}[ht!]
\[
\begin{array}{c}
\includegraphics[width=0.8\columnwidth]{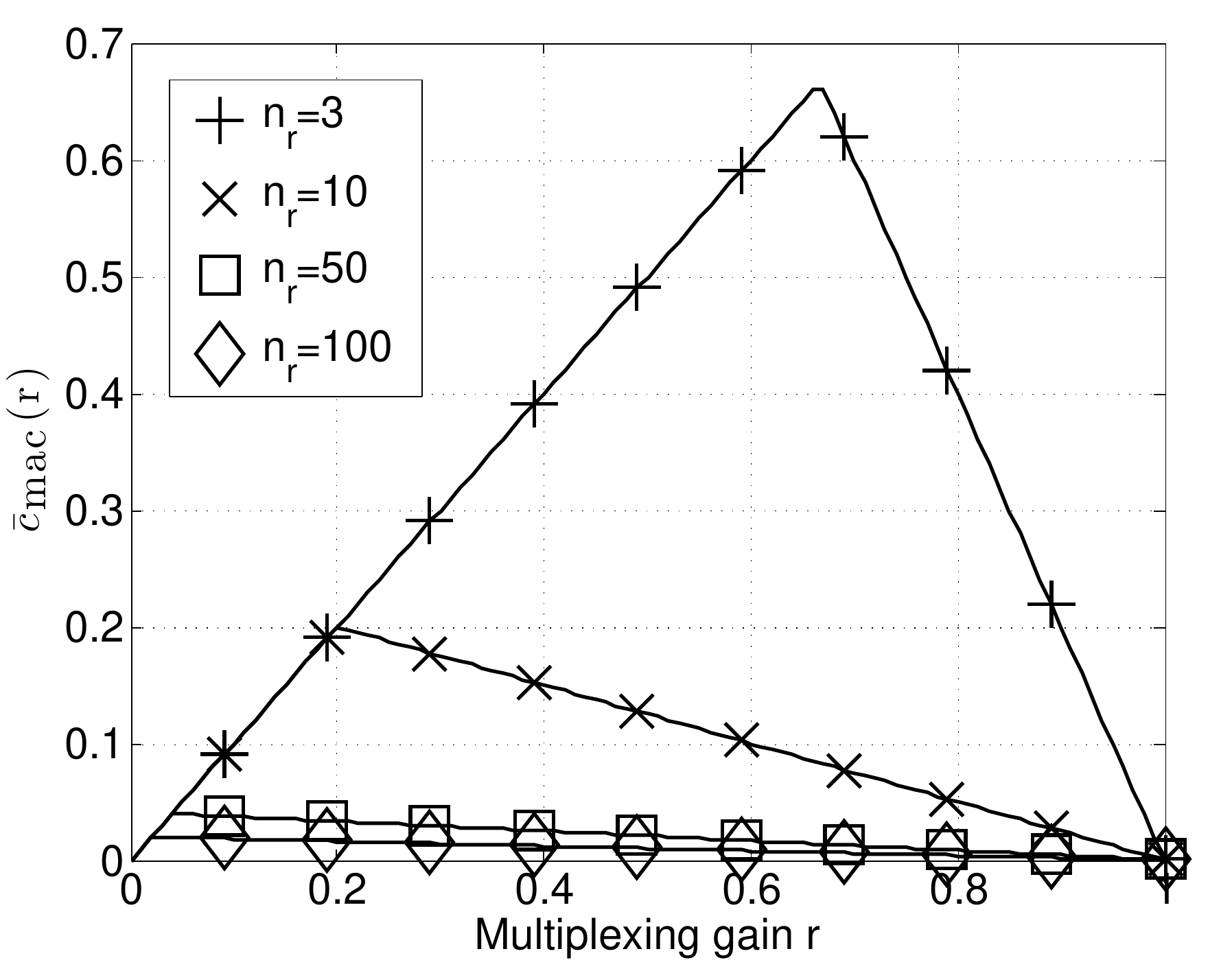}\\
\textnormal{(a)}\\
\includegraphics[width=0.8\columnwidth]{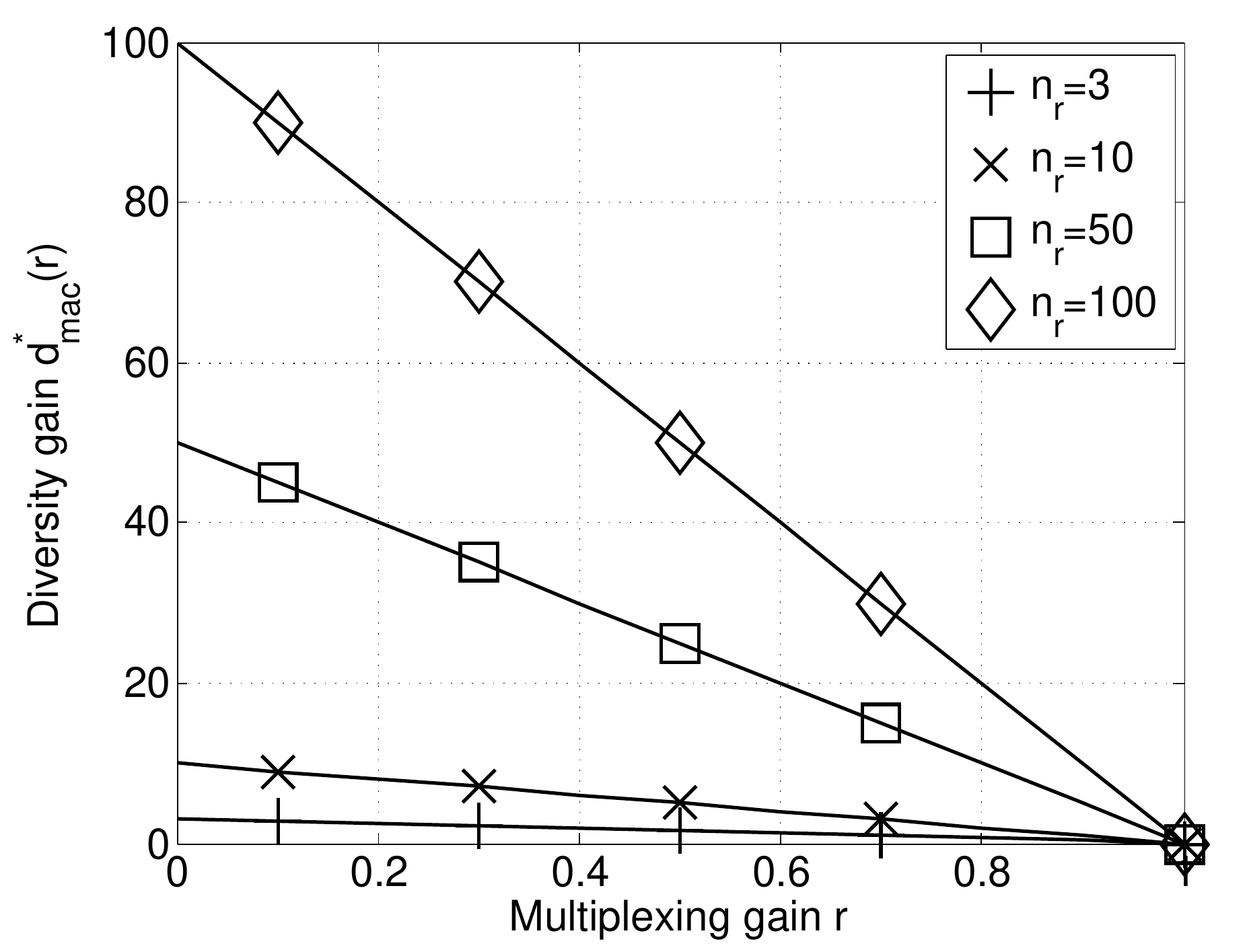}\\
\textnormal{(b)}
\end{array}
\]
\caption{Complexity-exponent upper bound $\bar{c}_\mac(r)$ (see subfigure (a)) for achieving $d^*_{\tm{mac}}(r)$ (see subfigure (b)) for $K=3$ and for increasing $n_r$.}
\label{fig:largenr}
\end{figure}

In Fig. \ref{fig:largenr}(a) we plot the complexity-exponent upper bounds $\bar{c}_\mac(r)$ for $K=3$ users and for various values of $n_r$. The corresponding optimal diversity gains $d^*_{\tm{mac}}(r)$ are shown in Fig. \ref{fig:largenr}(b). It can be easily seen that the bounds  $\bar{c}_\mac(r)$ (with $r$ fixed) decrease monotonically for an increasing $n_r$, and in the limit of very large $n_r$, the complexity exponent approaches $0$, meaning that the decoder ${\cal D}_r$ --- with the proper halting policy --- can deliver the substantial diversity performance $n_r(1-r)$ in a computationally inexpensive manner.

\vspace{3pt}

\section{Performance and Complexity for MAC ML-based Decoding with Feedback-Aided User Selection \label{sec:FB1}}
We now explore the ramifications of feedback-aided \emph{user-selection} on the performance and complexity of ML-based MAC communications. Our motivation in exploring user selection comes from the associated gains in MAC reliability which --- particularly for the high multiplexing-gain region --- is typically quite low. This same method simultaneously allows for a ``simplification" of the communication problem, from a larger multiple access channel, to a selectively reduced smaller and more manageable setting. Finally, this ability to provide simpler and more reliable communications, comes at a very reasonable feedback cost, which renders user selection applicable in the presence of limited feedback links.

We will here consider a MAC user-selection scheme, which is modified from the Jiang and Varanasi (JV) antenna selection algorithm in \cite{JiangVar09}. The JV algorithm, to the best of our knowledge, is currently the best performing antenna-selection algorithm for the point-to-point (single-pair) MIMO scenario. The proposed user selection scheme, to be described in more detail in Appendix~\ref{app:userSelectionScheme}, selects which $L$ out of $K$ users will transmit throughout the coherence period of the channel, and then informs all users of the selection outcome via a feedback channel using $\log_2 \binom{K}{L}$ bits per channel coherence period. This selection decision is taken as a function of the entries of a matrix $\bfR_{\tm{JV}}$ derived from the QR-decomposition $\bfH_{\tm{\!eq}} {\bf \Pi} = {\bf Q}_{\tm{JV}} \bfR_{\tm{JV}}$ where ${\bf \Pi}$ is a permutation matrix, where ${\bf Q}_\tm{JV}$ is a unitary matrix, and where --- depending on the values of $K$, $L$, and $n_r$ --- $\bfR_{\tm{JV}}$ can be either upper triangular or upper trapezoidal. This distinction in the shape of $\bfR_\tm{JV}$ will here limit our consideration of $L$ to an allowable set
\[\mathcal{L} := \{1,2,\cdots,\nu,K\}
\]
where we recall that $\nu = \min\{K,n_r\}$, and where $L=K$ corresponds to the case where no users are pruned out and thus where all $K$ users transmit simultaneously to the receiver. Such choice of $L=K$ can, especially at high values of multiplexing gains, provide for the highest reliability.

We proceed to bound the DMT gains of user selection, and --- based on these bounds --- to then show that any such gains can come with a \emph{simultaneous} exponentially-reduced complexity, compared to the complexity associated to $c^*_\mac(r)$ in Theorem~\ref{thm:cr_mac}.
\begin{note}
We note that under the adopted assumption that the fading statistics are uniform across users, then the user selection algorithm will in fact `prune’ all users (statistically) evenly and thus `fairly’. It is important to note that the rate considered here, does indeed account for the periods of `silence’ of a user that has not been selected. In other words, the derived diversity-multiplexing expressions --- which are the same for all users --- correspond to a user rate that is simply the total number of bits sent by each user, divided by the total number of time-slots including the time-slots during which the user was kept silent by the selection algorithm.
\end{note}


\subsection{Selection-aided MAC DMT bounds, for a fixed $L$}
%

We first proceed with an upper bound on the selection-aided DMT $d_{\tm{us},L}(r)$ for the employed user selection scheme, when the number of selected users is fixed to a certain $L$.
\begin{thm} \label{thm:dus}
Given $L$, and given the JV-inspired user selection scheme, the MAC DMT is upper bounded as
\begin{equation}
d_{\tm{us},L}(r)\leq \bar{d}_{\tm{us},L}(r) := \min_{k \geq 0, , \ell \geq 1 \atop k+\ell \leq L} d_{k,\ell}\left( \frac{\ell K r}{L}  \right)
\end{equation}
where
\begin{IEEEeqnarray*}{rCl}
d_{k,\ell}(r) &:=& \inf_{{\cal A}_\ell(r)} D_{k,\ell}(\bfalpha),\\
{\cal A}_\ell(r) &:=& \left\{ 0 \leq \alpha_1 \leq \cdots \leq \alpha_\ell : \sum_{i=1}^\ell (1-\alpha_i)^+ \leq r \right\}
\end{IEEEeqnarray*}
and
\begin{multline}
D_{k,\ell}(\bfalpha) := \sum_{i=1}^\ell  (n_r + \ell - 2i+1)  \alpha_i
+ \sum_{i=1}^{\ell-1} (K-k-\ell)\alpha_i \\
+\alpha_\ell (K-k-\ell)(n_r-k-\ell+1)\label{eq:pkell}.
\end{multline}
\end{thm}
\begin{proof}
See Appendix \ref{proof:dus}.
\end{proof}
\vspace{3pt}

We here note that it is not hard to show when $L=K$, that
\beq \label{eq:LisK}
d_{\tm{us},K}(r) = d^*_\mac(r) = \bar{d}_{\tm{us},K}(r) .
\eeq
Also when $n_r = 1$, it is easy to show that
\beq \label{eq:Lis1}
d_{\tm{us},1}(r) = \bar{d}_{\tm{us},1}(r).
\eeq
Hence the bound is tight for $L=1$ and $L=K$.


In Fig. \ref{fig:3} we plot the DMT upper bounds $\bar{d}_{\tm{us},L}(r)$ from Theorem \ref{thm:dus}, for the (under-determined) case of $K=4, n_r=3$, with $L=1,2,3$. In the same figure we compare the above $\bar{d}_{\tm{us},L}(r)$ to the optimal MAC DMT $d^*_\mac(r)$ corresponding to having no selection (or equivalently to having $L=K=4$).
Fig.~\ref{fig:4} plots the DMT bounds for the (over-determined) case of $K=3,n_r=4$, with $L=1,2,3$. In interpreting the figures, one must recall that, due to the fact that each user is selected with probability $\frac{L}{K}$, maintaining an average multiplexing gain $r$, will require that each transmitting user communicates at a multiplexing gain of $\frac{K}{L} r$.
This is indeed taken into account, and the multiplexing gain reflects the true rate of communication, in full consideration of the fact that part of the time, some users are not transmitting.  It should also be noted that for $L \leq \nu$, the maximal achievable multiplexing gain is $\frac{L}{K}$ or each user. Furthermore, it can be seen from the curves in Fig. \ref{fig:3} and Fig. \ref{fig:4} that the DMT performance of proposed feedback-aided user-selection scheme strongly depends on  the choice of  $L$ for different regions of multiplexing gain values $r$. For instance, in Fig. \ref{fig:4} of $K=3$ and $n_r=4$, choosing $L=1$ yields the largest diversity gain whenever $r \in (0, 0.2]$. For $r \in (0.2, 0.5]$, $L=2$ becomes the best choice, and if the desired multiplexing gain $r\in (0.5, 1]$, one should set $L=K=3$. A similar observation can be made in Fig. \ref{fig:3}. These cross-points between curves $\bar{d}_{\tm{us},L}(r)$ for various $L$ suggest further diversity gains by optimizing over $L$ as a function of $r$. This will be discussed in more detail in Section \ref{sec:III.c}.

Another interesting observation from Fig. \ref{fig:3} and Fig. \ref{fig:4} is that despite the different configurations of $K$ and $n_r$ in these figures, the respective initial values $\bar{d}_{\tm{us},L}(0)$ are the same for $L=1,2,3$. This is actually no coincidence. To see it, note that as multiplexing gain $r$ approaches zero from the right, meaning that the desired transmission rate $R$ is arbitrarily close to zero, the only way for the communication system being in outage is exactly when all the {\em major} channel links are in deep-fade, i.e., when $\alpha_1=\ldots=\alpha_\ell=1$ in the set ${\cal A}_\ell(r)$. In this case, the function $D_{k,\ell}(\bfalpha={\bf 1})$ simplifies to
\[
D_{k,\ell}({\bf 1}) = K n_r - k(K+n_r) + k^2 + k \ell,
\]
hence the minimal $D_{k,\ell}({\bf 1}) $ occurs at $\ell=1$ and
\[
\bar{d}_{\tm{us},L}(0) = \min_{0 \leq k \leq L-1} K n_r - k(K+n_r) + k^2 + k.
\]
This shows that the value of $\bar{d}_{\tm{us},L}(0)$ is symmetric between $K$ and $n_r$ whenever $L \leq \min\{K, n_r\}$. Moreover, with $K=4$ and $n_r=3$ (and the same holds for  $K=3$ and $n_r=4$) we have $\bar{d}_{\tm{us},L}(0) =  \min_{0 \leq k \leq L-1} k^2 - 6k + 12$, and it is straightforward to see that $\bar{d}_{\tm{us},1}(0)=12$, $\bar{d}_{\tm{us},2}(0)=7$, $\bar{d}_{\tm{us},3}(0)=4$, and $\bar{d}_{\tm{us},4}(0)=3$ as shown in Fig. \ref{fig:3}.



\begin{figure}[ht!]
\[
\includegraphics[width=0.8\columnwidth]{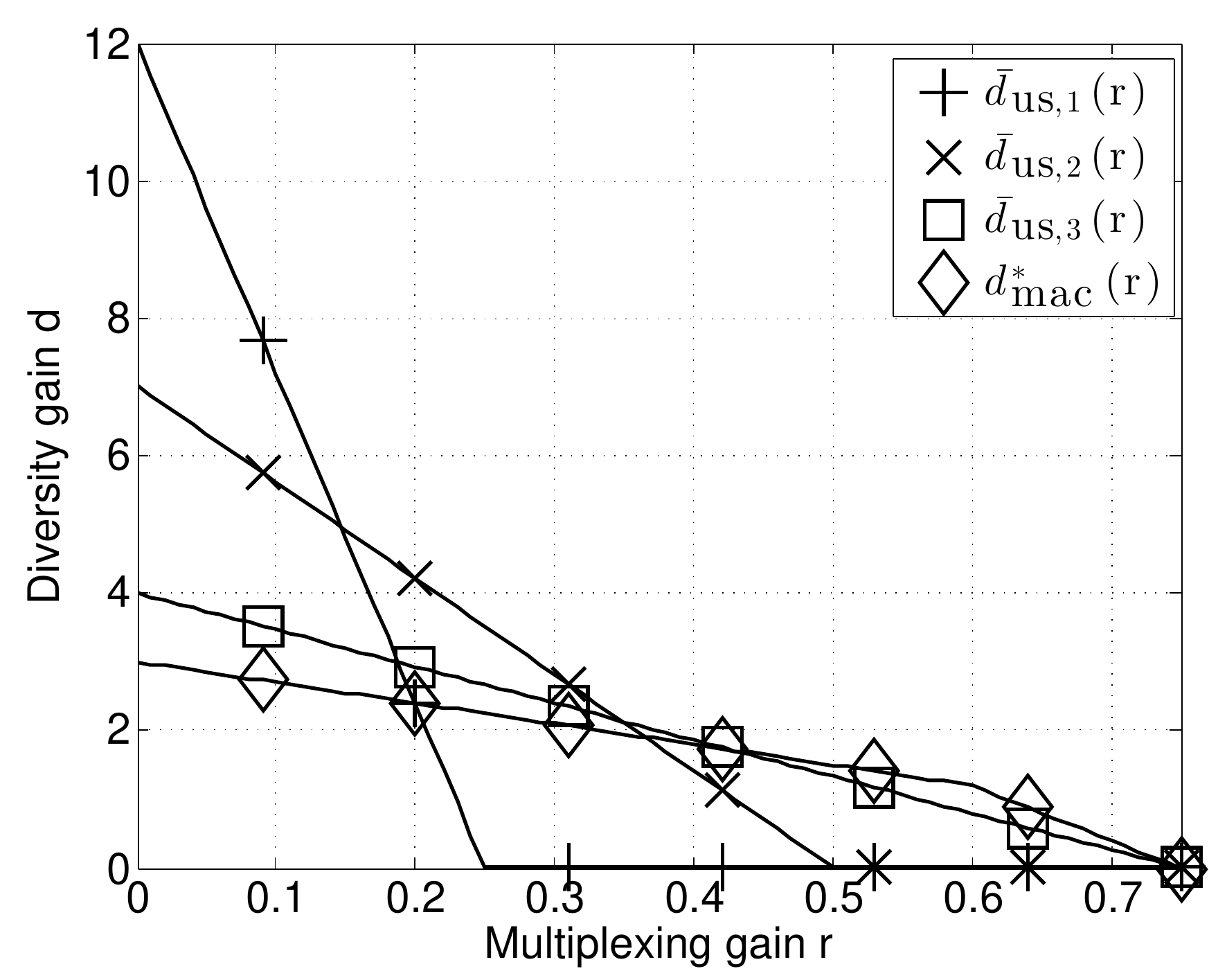}
\]
\caption{Selection-aided DMT upper bounds $\bar{d}_{\tm{us},L}(r)$, vs. the optimal MAC DMT $d^*_\mac(r)$ without selection. $K=4$, $n_r=3$, and $L=1,2,3$.} \label{fig:3}
\end{figure}

\begin{figure}[ht!]
\[
\includegraphics[width=0.8\columnwidth]{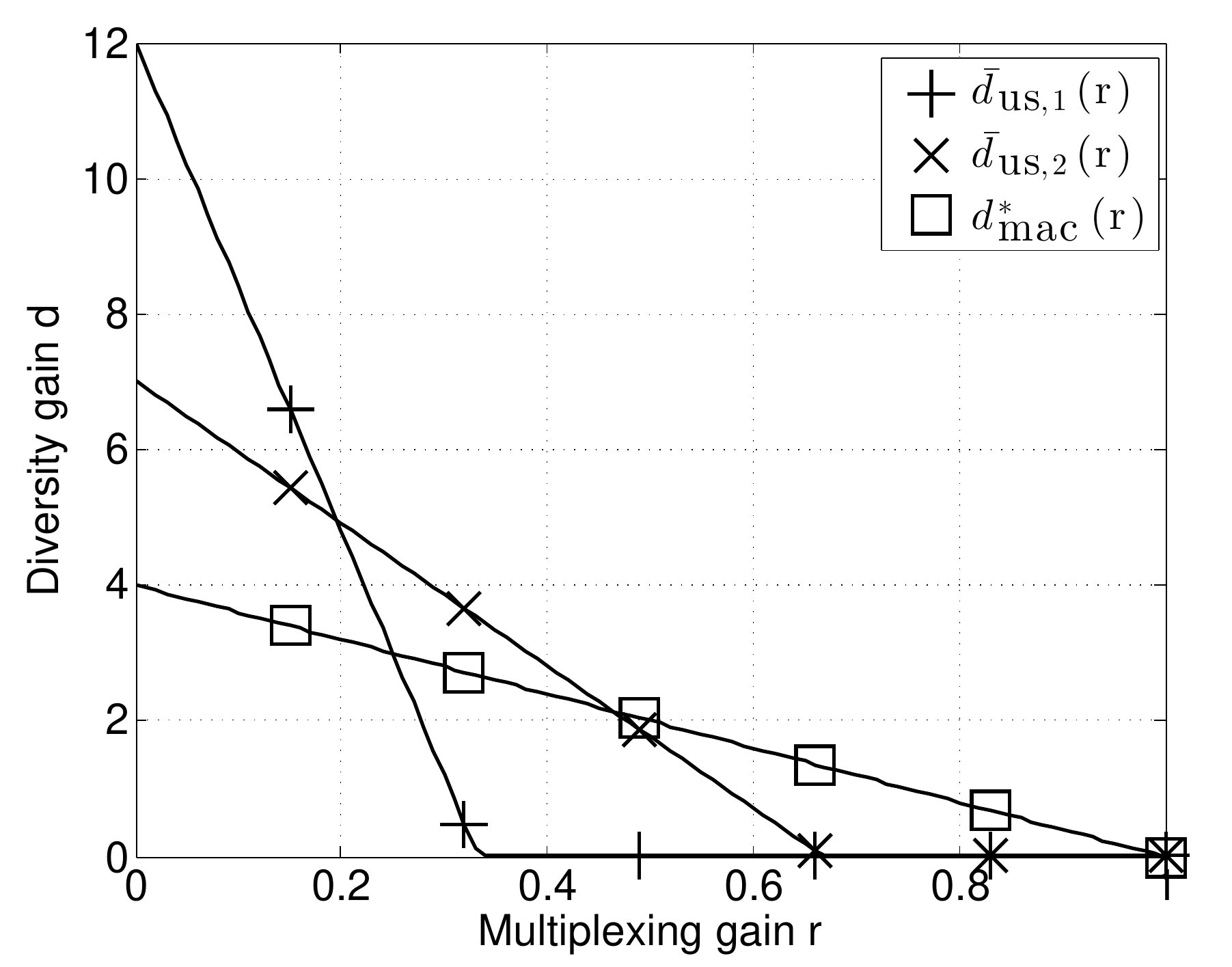}
\]
\caption{Selection-aided DMT upper bounds $\bar{d}_{\tm{us},L}(r)$, vs. the optimal MAC DMT $d^*_\mac(r)$ without selection. $K=3$, $n_r=4$, and $L=1,2$.} \label{fig:4}
\end{figure}

\subsection{Selection-aided MAC complexity bounds, for a fixed $L$}
We now use the previous DMT bound to upper bound the complexity exponent $c_{\tm{us},L}(r)$ required to achieve this selection-aided DMT $d_{\tm{us},L}(r)$ corresponding to the proposed $L$-user selection algorithm, for any fixed $L\in \mathcal{L}$. The result is a direct consequence of Theorem~\ref{thm:cr_mac}, and it holds for i.i.d. Rayleigh fading channel statistics.

\vspace{3pt}

\begin{thm}
\label{thm:cr_us}
For the $K$-user MAC, the complexity exponent $c_{\tm{us},L}(r)$ required to achieve the selection-aided DMT $d_{\tm{us},L}(r)$, for a given $L \leq \nu$, is upper bounded as
\beq
c_{\tm{us},L}(r) \leq \bar{c}_{\tm{us},L}(r)
=
\sup\limits_{\bfalpha \in {\cal F}(r)} \sum_{i=1}^{L}  \left[ \min\left\{ \frac{K}{L} r, \frac{K}{L} r+\alpha_i -1\right\}\right]^+
\label{eq:cusbar}
\eeq
where
\[
{\cal F}(r) := \left\{ \bfalpha:
\begin{array}{l}
\alpha_1 \leq \cdots \leq \alpha_L, 0 \leq \alpha_i \in \R\\
 D_{0,L} (\bfalpha)  \leq \bar{d}_{\tm{us},L}(r)
\end{array} \right\}
\]
and where $D_{k,\ell} (\bfalpha)$ is defined in~\eqref{eq:pkell}. The performance-complexity pair $(d_{\tm{us},L}(r),c_{\tm{us},L}(r))$ is achieved with uncoded QAM ($T=1$), a sphere decoder with a search radius $\delta > \sqrt{\bar{d}_{\tm{us},L}(r) \log \snr}$, any decoding ordering, and a decoding halting policy that halts decoding if $N_{\max}(r) \doteq \snr^{c_{\tm{us},L}(r)}$.
\end{thm}
\vspace{3pt}

\begin{proof}
The result holds directly from Theorem~\ref{thm:cr_mac} and from the fact that the complexity needed to achieve $d_{\tm{us},L}(r)$ is upper bounded by the complexity needed to achieve $\bar{d}_{\tm{us},L}(r) \geq d_{\tm{us},L}(r)$.
\end{proof}

\subsection{Bound-based performance and complexity optimization over $L$} \label{sec:III.c}
The performance and complexity bounds in Theorems~\ref{thm:dus} and \ref{thm:cr_us} suggest gains by optimizing over $L$ as a function of $r$.
Using the bounds as an indicator to the best choice of $L$, we let
\beq L^{*}(r) = \arg\max_{L\in \mathcal{L}}\{\bar{d}_{\tm{us},L}(r) \},
\eeq
where we recall that $\mathcal{L}= \{ 1,2,\ldots,\nu,K  \}.$
Hence
\beq
\bar{d}_{\tm{us},L^*}(r) = \max_{L\in \mathcal{L}}\{\bar{d}_{\tm{us},L}(r)\}
\eeq
can serve as an upper bound on the selection-aided DMT $\max_{L\in \mathcal{L}}\{d_{\tm{us},L}(r)\}$ maximized --- at any given $r$ --- over the choices of $L$, including over the choice $L=K$ that all users are selected.
As a result, going back to~\eqref{eq:cusbar}, we can calculate $\bar{c}_{\tm{us},L^*}(r)$ as a bound for the complexity required to achieve the optimized selection-aided DMT $\max_{L\in \mathcal{L}}\{d_{\tm{us},L}(r)\}$.

For the under-determined case of $K=4,n_r=3$, Fig.~\ref{fig:Lstar} plots $L^{*}(r)$, while Fig.~\ref{fig:6} plots the optimized, over $L$, selection-aided DMT upper bound $\bar{d}_{\tm{us},L^*}(r)$. This bound is also compared to the bound $\bar{d}_{\tm{us},L^{**}}(r)$ where
\beq
L^{**}(r) = \arg\max_{L\leq \nu}\{\bar{d}_{\tm{us},L}(r) \}
\eeq
corresponding to the case where selection always happens, i.e., corresponding to the case where we do not allow for $L=K$, even if $L^{*}(r) = K$.

Fig.~\ref{fig:7} plots the complexity exponent upper bound $\bar{c}_{\tm{us},L^*}(r)$ that guarantees the optimal --- over all $L$ --- selection-aided DMT $d_{\tm{us},L^*}(r)$. This complexity exponent bound is compared to $\bar{c}_\mac(r)$ that achieves $d^*_{\tm{mac}}(r)$ without user selection. The same figure also gives an interesting comparison to the complexity exponent $\bar{c}_{\tm{us},L^{**}}(r)$ which, in conjunction with $\bar{d}_{\tm{us},L^{**}}(r)$ in Fig.~\ref{fig:6}, suggests that the choice of $L^{**}(r)$  --- rather than $L^{*}(r)$ --- can provide substantial gains in complexity, with possibly moderate losses in performance.

%
%
%
%
%
%
%

\begin{figure}[ht!]
\[
\includegraphics[width=0.8\columnwidth]{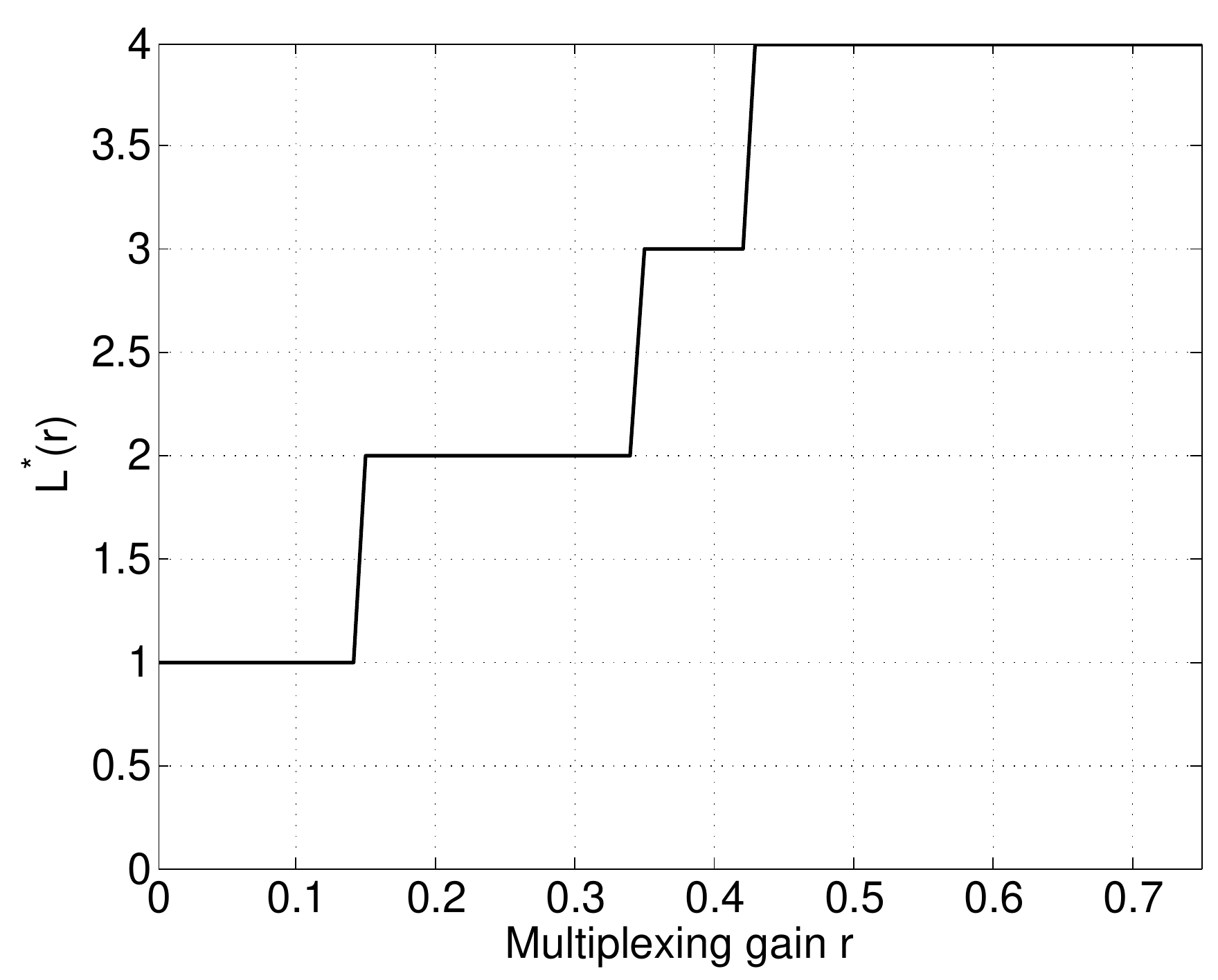}
\]
\caption{Optimal number of selected users $L^*(r)$ for $K=4$-user MAC with $n_r=3$.} \label{fig:Lstar}
\end{figure}


\begin{figure}[ht!]
\[
\includegraphics[width=0.8\columnwidth]{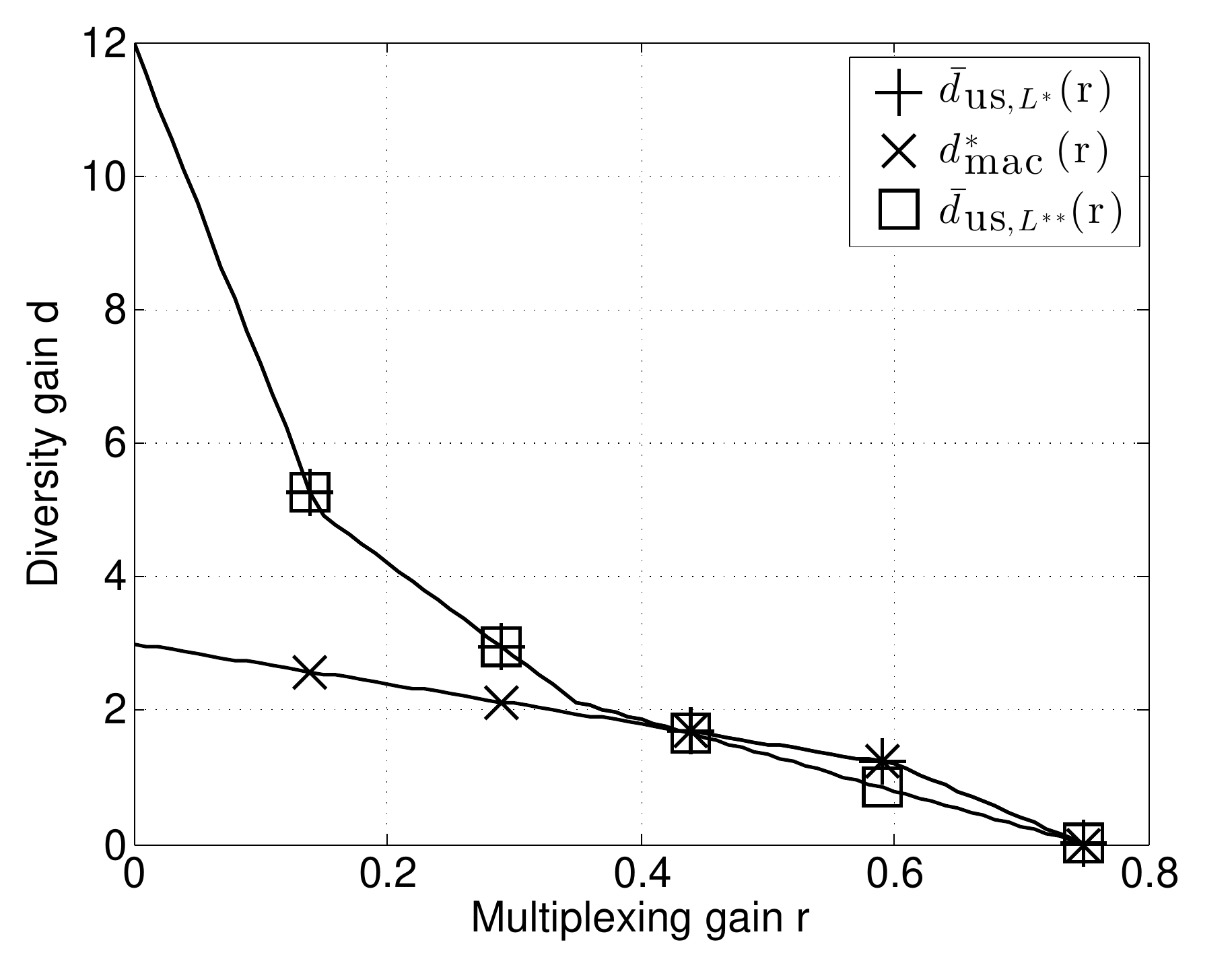}
\]
\caption{DMT upper bounds $\bar{d}_{\tm{us},L^*}(r)$ for user selection, compared to $d^*_\mac(r)$ without selection, and $\bar{d}_{\tm{us},L^{**}}(r)$ when user-selection is enforced, for the $K=4$-user MAC with $n_r=3$.} \label{fig:6}
\end{figure}

\begin{figure}[ht!]
\[
\includegraphics[width=0.8\columnwidth]{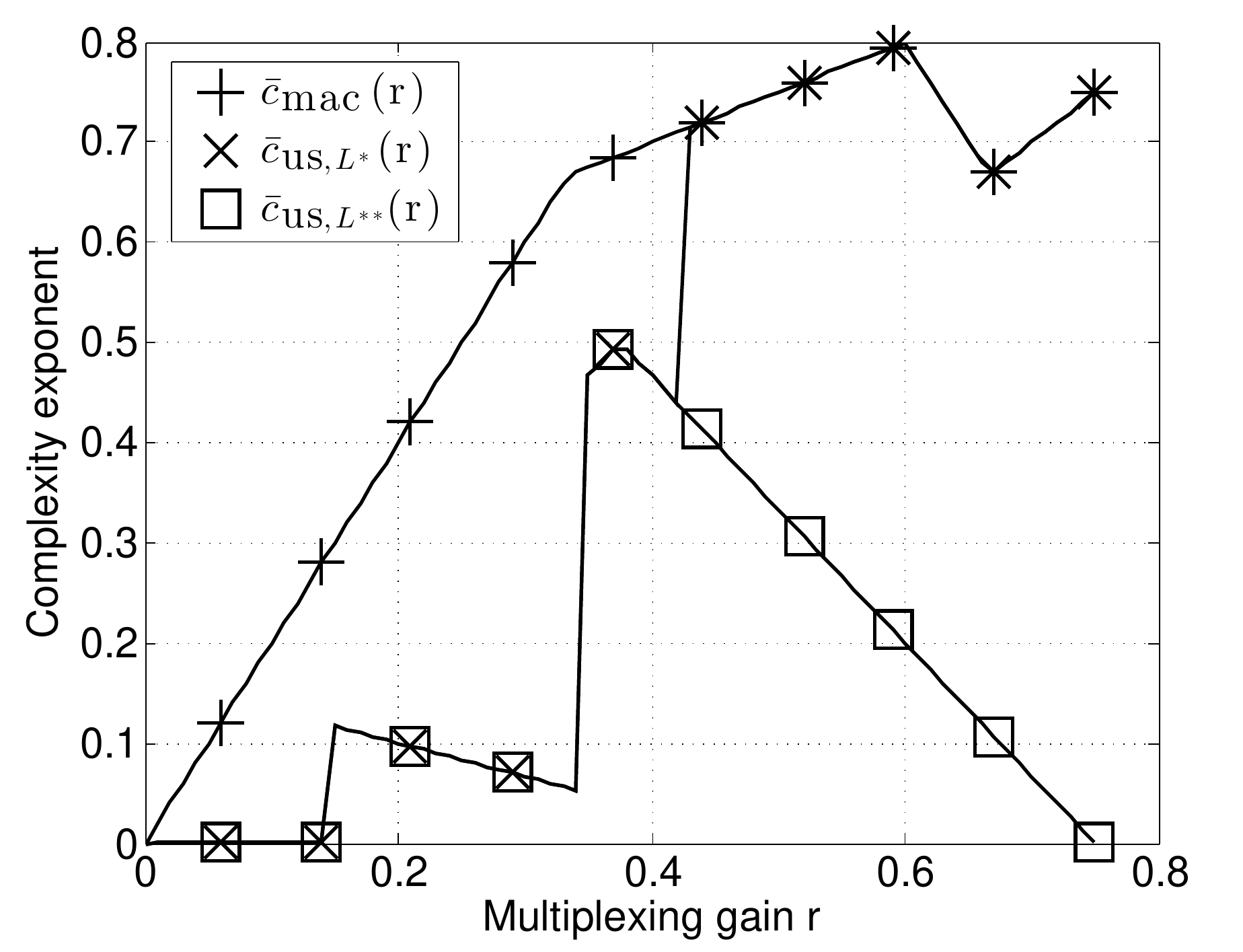}
\]
\caption{Complexity exponent upper bounds $\bar{c}_{\tm{us},L^*}(r)$ for user selection, compared to $\bar{c}_\mac(r)$ without selection, and compared to $\bar{c}_{\tm{us},L^{**}}(r)$ corresponding to when user-selection is enforced. Plots for $K=4$ and $n_r=3$.} \label{fig:7}
\end{figure}

%
%
%

\subsection{Selection-aided complexity reduction, at no performance costs}
We here explore the particular case where user selection is calibrated to reduce complexity, while maintaining reliability. The following quantifies this approach, based on the derived bounds. A crucial element in this effort is that any gains in the selection-aided DMT compared to the original MAC-DMT $d_{\tm{mac}}^*(r)$, can allow us to effectively `rest' the decoder often enough to reduce the complexity exponent, without falling below the required $d_{\tm{mac}}^*(r)$. In this case, the resulting complexity exponent is upper bounded by
\begin{eqnarray}
\lefteqn{\bar{c}_{\tm{red-us},L}(r)}\nonumber\\
&=& \left\{
\begin{array}{l}
\sup\limits_{\bfalpha \in {\cal F}_{\tm{red-us},L}(r)} \sum_{i=1}^L \left[ \min \left\{ \frac{K}{L}r, \frac{K}{L}r + \alpha_i - 1 \right\} \right]^+, \\
\hfill \tm{if } \bar{d}_{\tm{us},L}(r) \geq d_{\tm{mac}}^*(r), \\\\
\bar{c}_{\tm{mac}}(r),
\hfill \tm{ if otherwise},
\end{array} \right.\nonumber\\ \label{eq:credusL}
\end{eqnarray}
where
\[
{\cal F}_{\tm{red-us},L}(r) \ = \ \left\{ \bfalpha : \begin{array}{l}
\alpha_1 \leq \alpha_2 \leq \cdots \leq \alpha_L, 0 \leq \alpha_i \in \R\\
D_{0,L}(\bfalpha) \leq d_{\tm{mac}}^*(r)
\end{array} \right\}
\]
and where the second case in~\eqref{eq:credusL} refers to the situation where $L^* = K$. Hence an upper bound on the associated complexity exponent, minimized over $L$, is given by
\beq
\bar{c}_{\tm{red-us}}(r) \ = \ \min_{L \in {\cal L}} \bar{c}_{\tm{red-us},L}(r).
\eeq
Fig.~\ref{fig:8} explores this for the case of $K=4$ and $n_r=3$, presenting the resulting upper bounds $\bar{c}_{\tm{red-us},L}(r)$ for all possible $L=1,2,3$. The optimized complexity reduction, under the condition that we do not reduce reliability below $d_{\tm{mac}}^*(r)$, is given in Fig.~\ref{fig:9}. The striking observation is that at low or moderate multiplexing gain regimes, the optimal $d_{\tm{mac}}^*(r)$ can be achieved with a complexity exponent that can be as small as $0$, i.e., with computational complexity that does not scale exponentially in the total number of codeword bits. At high multiplexing gains, no complexity reduction is available since $L^* = K$.

\begin{figure}
\[
\includegraphics[width=0.8\columnwidth]{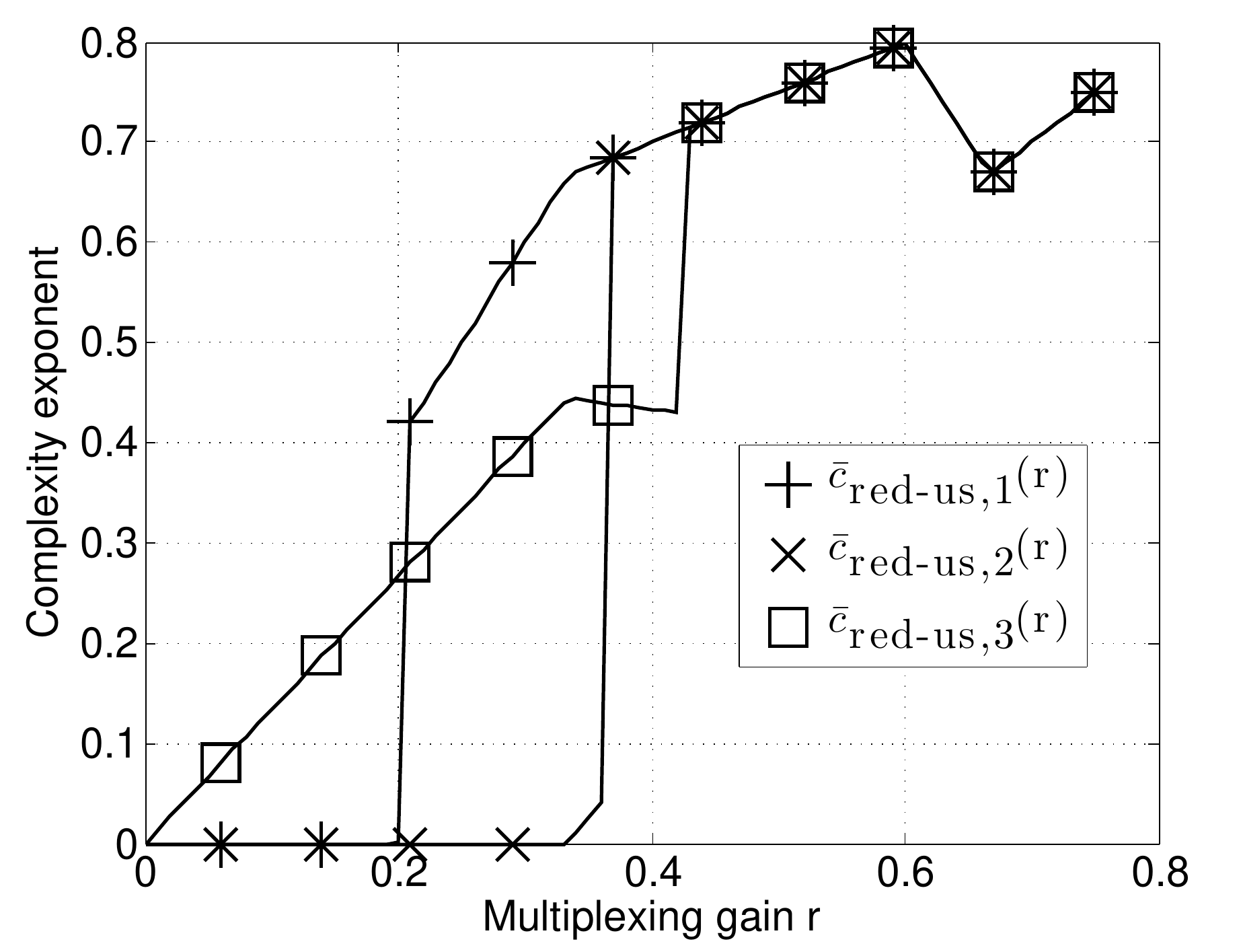}
\]
\caption{Complexity exponent bounds $\bar{c}_{\tm{red-us},L}(r)$ with $L=1,2,3$, all  guaranteeing $d_{\tm{mac}}^*(r)$ performance for $K=4$ ,and $n_r=3$. } \label{fig:8}
\end{figure}

\begin{figure}
\[
\includegraphics[width=0.8\columnwidth]{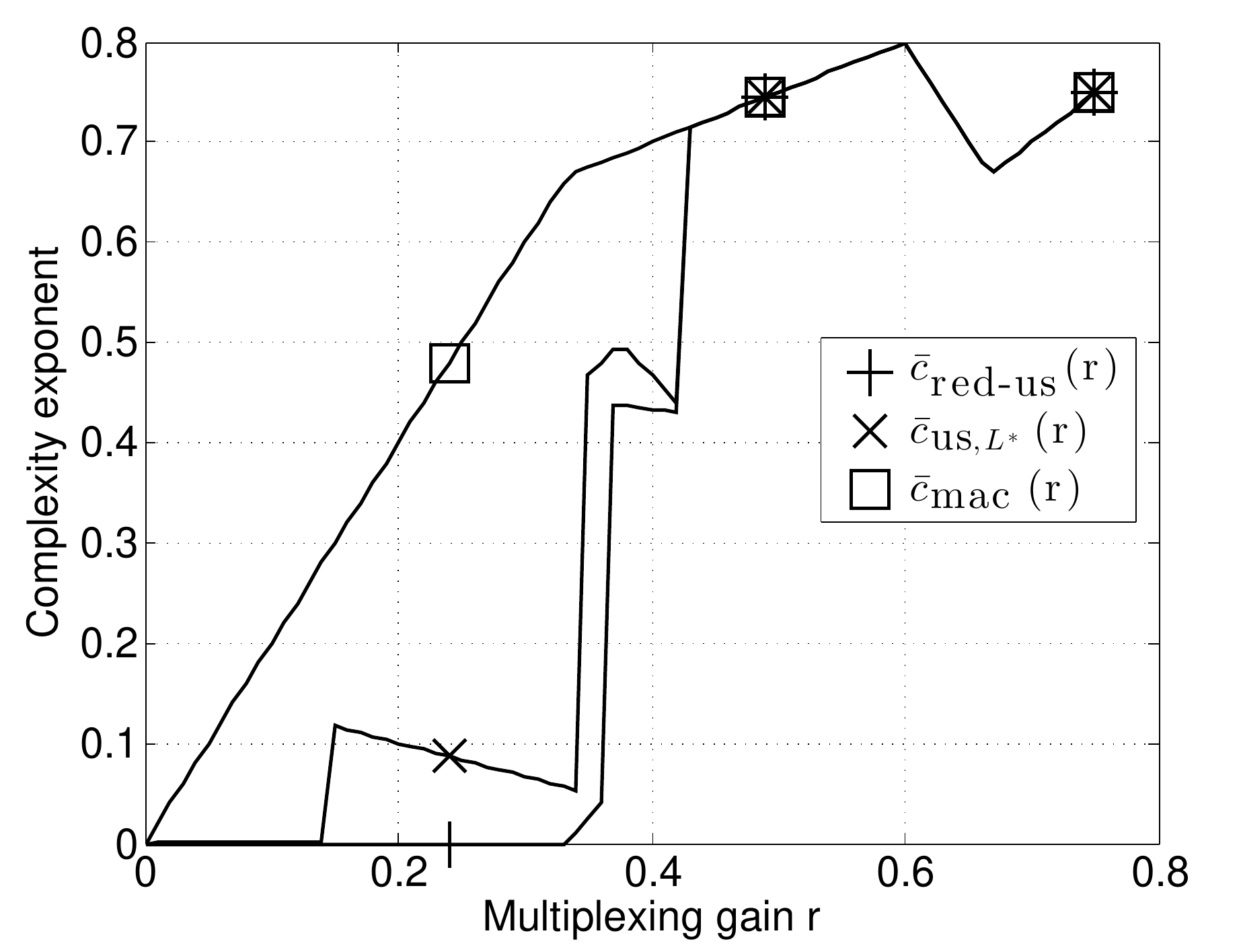}
\]
\caption{Complexity exponent bounds $\bar{c}_{\tm{red-us}}(r)$, $\bar{c}_{\tm{us},L^*}(r)$ and $\bar{c}_\mac(r)$, all guaranteeing $d_{\tm{mac}}^*(r)$ performance for $K=4,n_r=3$.} \label{fig:9}
\end{figure}

\section{Conclusions}
In this paper we derived bounds on the performance and complexity behavior of ML-based (and lattice-based) decoding, for the multiuser multiple access channel. Emphasis on ML-based decoding was motivated by the fact that MAC-related settings have distinctively reduced reliability, which hinders the use of low-performance decoders that would further reduce this reliability.

The derived complexity-vs-performance bounds were presented in the form of diversity and complexity exponents, and can provide insight on how to tradeoff performance and complexity in such outage-limited multiuser settings. The analysis shows that, under the requirement of efficient ML-based decoding, this tradeoff is not trivial; indeed we see that complexity constraints can crucially deteriorate performance.

The derived bounds also suggest substantial reliability and complexity benefits by increasing the number of receive antennas, as well as substantial benefits by utilizing just a few bits of feedback to allow for user selection. User selection is indeed pertinent as it can inherently increase reliability through channel selection, as well as can inherently decrease complexity by simply and selectively reducing the size of the problem at hand. These gains are highly sought in multiuser settings like the MAC, which typically suffer from reduced reliability and increased complexity of decoding.
In this context the analysis reveals the interesting finding that proper calibration of user selection can reduce the complexity exponent of near-optimal ML-based decoding, down to zero, thus revealing that --- for a substantial range of multiplexing gains --- the computational complexity of near-optimal ML-based decoding, need not scale exponentially in the total number of codeword bits.

\appendices
\section{Proof of Theorem~\ref{thm:mac_cr_for_optimalDMT}} \label{app:thm_cr_mac}
In the following we establish an upper bound on the minimum complexity exponent $c_{\tm{mac},d}(r)$ required by ML-based decoding to achieve a MAC DMT performance $d(r)$.

We recall from \eqref{eq:newvect_mac} that the sphere decoder `sees' a channel model
\begin{equation}
\bfy = \bfM \bfs + \bfw \label{eq:newvect_mac2}
\end{equation}
where $\bfM = \snr^{\frac {1-r}{2}} \bfH \bfG$, where $\bfH$ is the equivalent channel in~\eqref{eq:H}, and where $\bfG$ is the generator matrix of the overall (product) lattice code. To allow for sphere decoding in the under-determined ($K > n_r$), MMSE-preprocessing (cf.~\cite{conf_DaElCa04}) gives the MMSE-preprocessed code-channel matrix
\begin{equation}\label{eq:QRdecomposition}
\tilde{\bfM} = \left[
\begin{array}{c}
\bfM\\
\alpha \boldsymbol{I}_u
\end{array} \right] = \boldsymbol{QR}
\end{equation}
where $\boldsymbol{I}_u$ is the $u\times u$ identity matrix, where $u=2(K-n_r)T$, and where $\alpha=\snr^{- \frac r2}$. If $K \leq n_r$, then simply $\alpha=u=0$.
After the QR decomposition in~\eqref{eq:QRdecomposition}, we get
\[
\bfr \ := \ \left( \bfR^\dag \right)^{-1} \bfM^\dag \bfy = \bfR \bfs + \bfw'
\]
where $\bfw' = -\alpha^2 \left( \bfR^\dag \right)^{-1} \bfs + \left( \bfR^\dag \right)^{-1} \bfM^\dag \bfw$. For \[{\mathbb S}_r := \left\{ \bfs \in \Z^{2KT} : \snr^{-\frac r2} \bfG \bfs \in {\cal X}_r\right\}\] being the set of points in the product lattice $\Lambda$ that constitute ${\cal X}_r$ after scaling, the sphere decoder takes the form
\begin{equation}
\hat{\bfs}_{\tm{MMSE-SD}} \ = \ \arg \min_{\bfs \in {\mathbb S}_r} \norm{\bfr - \bfR \bfs}^2
\end{equation}
and it recursively enumerates all candidate vectors $\bfs \in {\mathbb S}_r$ within a given search sphere of radius $\delta = \sqrt{z \log \snr}$ for some $z > d(r)$. 

To compute an upper bound on the complexity exponent, we follow the approach similar to \cite{SinghEJ12}. Towards this, we first let \[\lambda_i=\sigma_i (\bfH_{\tm{\!eq}}^\dag \bfH_{\tm{\!eq}}), \ \ i=1, \ldots, \nu=\min\{ K, n_r\}\] be the nonzero singular values of $\bfH_{\tm{\!eq}}^\dag \bfH_{\tm{\!eq}}$, arranged in ascending order, and then we let
\beq \label{eq:mu}\mu_i = - \frac{\log \lambda_i}{\log \snr}.\eeq
Hence we can write
 \begin{equation}
 \sigma_i(\bfR) \ = \  \sigma_i(\tilde{\bfM}) = \sqrt{\alpha^2 + \sigma_i(\bfM^\dag \bfM)}, \ \ i=1, \ldots, 2KT. \label{eq:mu_heq_ml_maco}
 \end{equation}
 For $K \leq n_r$, we have
\[ \sigma_i(\bfR) \doteq \snr^{\frac 12 (1-r-\mu_{\lceil \frac{i}{2T}\rceil})}\]
and for $K > n_r$ we have
\begin{equation}
\sigma_i(\bfR) \doteq \left\{
\begin{array}{ll}
\snr^{- \frac r2}, & \tm{ if $1 \leq i \leq 2T(K-n_r)$}\\
\snr^{-\frac r2 + \frac 12 (1- \mu_j)^+}, & \tm{ otherwise}
\end{array} \right\} \label{eq:mu_heq_ml_macu}
\end{equation}
where $j = \lceil \frac{i-2T(K-n_r)}{2T}\rceil$. In the above we have used that $\sigma_{i}(\bfG) \doteq \snr^0$, for all $i$.

We now see that for any given channel realization, corresponding to a specific $\boldsymbol{\mu} :=[\mu_1 \cdots \mu_{\nu}]^\top$ (cf.\eqref{eq:mu}), the total number of visited nodes is given by
\[
N_{\tm{SD}}(\boldsymbol{\mu}) := \sum_{k=1}^{2KT} N_k (\boldsymbol{\mu})
\]
where, drawing from \cite[Lemma 1]{JaldenE12}, we can see that
\[
N_k (\boldsymbol{\mu}) \ \leq \
\prod_{i=1}^k \left[ \sqrt{k} + 2 \min \left\{ \frac{\delta}{\sigma_i(\bfR)}, \sqrt{k} \snr^{\frac r2}\right\} \right]
\]
and thus that
\[
N_{\tm{SD}}(\boldsymbol{\mu})  \leq  \sum_{k=1}^{2KT} \prod_{i=1}^k \left[ \sqrt{k} + 2 \min \left\{ \frac{\delta}{\sigma_i(\bfR)}, \sqrt{k} \snr^{\frac r2}\right\} \right].
\]
Thus for the overdetermined case (cf.~\eqref{eq:mu_heq_ml_maco}), we have
\begin{equation}
N_{\tm{SD}}(\boldsymbol{\mu}) \ \dot\leq \ \snr^{T \sum_{i=1}^K \left[ \min\left\{ r, r+\mu_i-1\right\}\right]^+} \label{eq:nsdo}
\end{equation}
while for the underdetermined case (cf.~\eqref{eq:mu_heq_ml_macu}) we have
\begin{equation}
N_{\tm{SD}}(\boldsymbol{\mu}) \ \dot\leq \ \snr^{(K - n_r)rT+T\sum_{i=1}^{n_r}\left(r-(1-\mu_{i})^+\right)^+} \label{eq:nsdu}
\end{equation}
where we have used the $2T$-fold multiplicity of the singular values of $\bfH$.

At this point, in the same spirit as in~\cite{SinghEJ12}, the upper bound on the complexity exponent can be obtained as the solution to a constrained minimization problem of finding a value $c_{\tm{mac},d}(r)$ such that the probability of a premature termination of SD algorithm is no larger than the channel outage probability, i.e.,
\begin{equation}
\Pr\left\{ N_{\tm{SD}}(\boldsymbol{\mu}) \geq N_{\max}(r)=\snr^{c_{\tm{mac},d}(r)}\right\} \leq \snr^{-d(r)}.  \label{eq:nsdout}
\end{equation}
This is the optimization reflected in the complexity exponent bound of~Theorem \ref{thm:mac_cr_for_optimalDMT}.\qed

\section{Proof of Theorem \ref{thm:qamau}: Optimality of Uncoded QAM} \label{app:qamau}

For $T=1$, the rank-$2K$ overall (product) code-lattice $\Lambda$ is isomorphic to the rectangular lattice $\Z^{2K}$, and the overall (product) code
\[
{\cal X}_r \ = \  \left\{ \snr^{-\frac r2} \bfx : \bfx \in \left( \Z[\im]\right)^K, \abs{x_i}^2 \leq \snr^r \right\}
\]
is essentially uncoded (scaled) QAM. To show that the above code ${\cal X}_r$ satisfies the probabilistic complexity constraint \eqref{eq:nsdout}, we follow the footsteps in \cite{Kuser, remarkdmt} and consider a $K$-fold extension of ${\cal X}_r$
\[
{\cal X}_{r,\tm{ext}} = \bigoplus_{i=1}^K {\cal X}_r \subset \snr^{- \frac r2}M_K(\Z[\im])
\]
where $M_K(\Z[\im])$ denotes the set of $K \times K$ matrices with entries from $\Z[\im]$. In other words, elements of ${\cal X}_{r,\tm{ext}}$ are square $K\times K$ matrices whose entries are independent QAM constellation points after a scaling of $\snr^{-\frac r2}$. Naturally, for any given decoder, ${\cal X}_{r,\ext}$ and ${\cal X}_r$ achieve the same DMT performance.  Under brute force ML decoding, the error probability of ${\cal X}_{r,\ext}$ is upper bounded by
\begin{IEEEeqnarray*}{rCl}
\lefteqn{P_{e,\ext}(r)}\\
&=& \E\Pr\left\{ \bfX' \in {\cal E}(\bfX) \tm{ decoded}\right\}\\
&=& \E \Pr\left\{ \bigcup_{k=1}^K \left\{ \bfX' \in {\cal E}_k(\bfX) \tm{ decoded}\right\}\right\}\\
&\stackrel{\tm{(i)}}{\leq} & \sum_{k=1}^K \E \Pr\left\{ \bfX' \in {\cal E}_k(\bfX) \tm{ decoded}\right\}\\
&\stackrel{\tm{(ii)}}{\leq} & \sum_{k=1}^K \E \Pr\left\{ \bigcup_{\bfX' \in {\cal E}_k(\bfX)} \biggl\{ \bfH_{\tm{\!eq}} :  \right.\\
&& \qquad \qquad \qquad \left. \norm{\snr^{\frac 12} \bfH_{\tm{\!eq}} (\bfX-\bfX')}^2 \dot\leq 1 \right\}\Biggr\}
\yesnumber \label{eq:Pe19}
\end{IEEEeqnarray*}
where the expectation is taken over all codeword matrices $\bfX \in {\cal X}_{r,\ext}$, where ${\cal E}(\bfX):={\cal X}_{r,\ext}\setminus\{\bfX\}$ is the set of all possible erroneous decoded outputs given that codeword matrix $\bfX$ is transmitted, where ${\cal E}_k(\bfX):=\left\{ \bfX' \in {\cal E}(\bfX) : \rank(\bfX-\bfX')=k\right\}$ with $k=1, \ldots, K$ forms a partition of ${\cal E}(\bfX)$, where step (i) follows from the union bound, and where step (ii) is due to the use of a suboptimal bounded distance decoder (cf. \cite{Kuser}).

For any $\bfX' \in {\cal E}_k(\bfX)$, set $\boldsymbol{\Delta}_{\bfX'} = \bfX - \bfX'$ and let $\boldsymbol{\Delta}_{\bfX'} \boldsymbol{\Delta}_{\bfX'}^\dag = \boldsymbol{U}_{\bfX'} \boldsymbol{\Sigma}_{{\bfX'}} \boldsymbol{U}_{\bfX'}^\dag$ be the corresponding eigen-decomposition. Note that as $\rank(\boldsymbol{\Delta}_{\bfX'}) = k$, the eigenvalue matrix $\boldsymbol{\Sigma}_{\bfX'}$ has form $\boldsymbol{\Sigma}_{\bfX'}=\tm{diag}(\boldsymbol{\Omega}_{\bfX'},\boldsymbol{0}_{K-k})$, where the diagonal of $\boldsymbol{\Omega}_{\bfX'}$ consists of all nonzero eigenvalues of $\boldsymbol{\Delta}_{\bfX'}\boldsymbol{\Delta}_{\bfX'}^\dag$. Substituting the above into \eqref{eq:Pe19}, we obtain
\begin{multline}
\Pr\left\{ \bigcup_{\bfX' \in {\cal E}_k(\bfX)} \left\{ \bfH_{\tm{\!eq}} : \norm{\snr^{\frac 12} \bfH_{\tm{\!eq}} \boldsymbol{\Delta}_{\bfX'}}^2 \dot\leq 1 \right\}\right\}\\
=\Pr\left\{ \bigcup_{\bfX' \in {\cal E}_k(\bfX)} \left\{ \bfG_k : \snr \cdot \trace\left(\bfG_k^\dag \boldsymbol{\Omega}_{\bfX'}\bfG_k\right) \dot\leq 1 \right\}\right\}
\end{multline}
where $\bfG_k$ is an $(n_r \times k)$ random matrix with i.i.d. $\CN{0}{1}$ entries. Let $\lambda_1 \leq \cdots \leq \lambda_m$ be the ordered nonzero eigenvalues of $\bfG_k^\dag \bfG_k$, where $m=\min\{k,n_r\}$. Noting that $\det(\boldsymbol{\Omega}_{\bfX'})\, \dot\geq\, \snr^{-kr}$ and that $\trace(\boldsymbol{\Omega}_{\bfX'})\, \dot\leq \, 1$, it can be shown --- drawing from \cite{Kuser} --- that the condition of $\trace\left(\bfG_k^\dag \boldsymbol{\Omega}_{\bfX'}\bfG_k\right) \, \dot\leq\,  1$ implies that $\sum_{i=1}^m (1-\mu_i)^+ \leq kr$, which is independent of the choice of $\bfX'$. Hence we have
\begin{IEEEeqnarray*}{rCl}
\lefteqn{\Pr\left\{ \bigcup_{\bfX' \in {\cal E}_k(\bfX)} \left\{ \bfG_k : \snr \cdot \trace\left(\bfG_k^\dag \boldsymbol{\Omega}_{\bfX'}\bfG_k\right) \dot\leq 1 \right\}\right\}}\\
&\dot\leq& \Pr\left\{ \bigcup_{\bfX' \in {\cal E}_k(\bfX)} \left\{ \boldsymbol{\mu} : \sum_{i=1}^m (1-\mu_i)^+ \leq kr\right\}\right\}\\
&=& \Pr\left\{ \boldsymbol{\mu} : \sum_{i=1}^m (1-\mu_i)^+ \leq kr\right\}\\
& \doteq & \snr^{-d_{k,n_r}^*(kr)}
\end{IEEEeqnarray*}
where the last exponential equality follows from \cite{ZheTse}. Finally, note that the error probability of ${\cal X}_r$, subject to (unterminated) joint ML decoding, is upper bounded by
\[
P_e(r) \leq \frac{1}{K} P_{e,\ext}(r) \,\dot\leq\, \frac{1}{K} \sum_{k=1}^K \snr^{-d_{k,n_r}^*(kr)} \, \doteq \, \snr^{-d^*_\mac(r)},
\]
which completes the proof.\qed

\section{Proof of Corollary~\ref{cor:largenr} \label{app:largenr}}

First let us recall that for $n_r > K$, then $\nu=\min\{K,n_r\}=K$. Let us also recall from \cite{ZheTse} that the joint probability density function of $\boldsymbol{\mu}=[\mu_1\, \cdots\, \mu_K]^\top$ (cf.~\eqref{eq:mu}) satisfies
\[
p \left( \boldsymbol{\mu}\right) \doteq \snr^{-\sum_{i=1}^K (n_r-K+2i-1) \mu_i}
\]
provided that $\mu_i \geq 0$ for all $i$. As a result, for each $\boldsymbol{\mu} \in (\R^+)^K$, the corresponding complexity exponent is upper bounded by
\beq \label{eq:cr upper bound app}
\bar{c}\left( \boldsymbol{\mu}\right) = \sum_{i=1}^K \left[ \min\{r,r+\mu_i-1\}\right]^+
\eeq
which is \emph{not a function} of $n_r$. At this point, let us consider a decoder that decodes only when
\[
\sum_{i=1}^K (n_r-K+2i-1) \mu_i < d_{1,n_r}^*(r)
\]
while when $\sum_{i=1}^K (n_r-K+2i-1) \mu_i > d_{1,n_r}^*(r)$ the decoder simply declares an error. Since, for $n_r>K$, $d_{1,n_r}^*(r)$ dominates the MAC DMT (cf.~\cite{TseVisZhe}), the extra declared errors do not affect the overall diversity performance.

Now we note that for any $\mu_1\geq \cdots \geq \mu_K \geq 0$ such that $\sum_i (n_r-K+2i-1) \mu_i = n_r(1-r)$, we have that
\[
\sum_i (n_r-K+1) \mu_i \leq n_r (1-r) \leq \sum_i (n_r-K+2K-1) \mu_i,
\]
and thus that
\beq
\frac{n_r}{n_r+K-1}(1-r) \leq \sum_i \mu_i \leq \frac{n_r}{n_r-K+1} (1-r) \label{eq:applarnr1}
\eeq
which, for $n_r \gg K$, implies that
\[
\frac{n_r}{n_r+K-1}(1-r) \approx \frac{n_r}{n_r-K+1} (1-r) \approx (1-r).
\]
Subject to the constraint in~\eqref{eq:applarnr1}, we can see that the maximum value of $\bar{c}(\boldsymbol{\mu})$ in~\eqref{eq:cr upper bound app}, is achieved by setting $\mu_1=\frac{n_r}{n_r-K+1} (1-r)$ and $\mu_2=\cdots = \mu_K=0$, where this maximal value takes the form
\[
\sum_{i=1}^K \left[ \min\{r,r+\mu_i-1\}\right]^+ = \min \left\{ r, \frac{K-1}{n_r-K+1}(1-r)\right\}.
\]
Hence for $n_r \gg K$,
\bea
\bar{c}_\mac(r) &=& \sup\limits_{\boldsymbol{\mu} \in {\cal B}(r)} \sum_{i=1}^{K}  \left[ \min\left\{ r, r+\mu_i -1\right\}\right]^+\nonumber\\
&\leq & \sup\limits_{\boldsymbol{\mu} \in \overline{{\cal B}}(r)} \sum_{i=1}^{K}  \left[ \min\left\{ r, r+\mu_i -1\right\}\right]^+ \label{eq:applargenr2}\\
&=& \min \left\{ r, \frac{K-1}{n_r-K+1}(1-r)^+\right\}\nonumber
\eea
where
\[
{\cal B}(r) := \left\{ \boldsymbol{\mu}:
\begin{array}{l}
\mu_1 \geq \cdots \geq \mu_{K}, 0 \leq \mu_i \in \R,\\
\ \sum_{i=1}^{K} \left( K-n_r+2i-1\right)\mu_i \leq n_r(1-r)^+
\end{array} \right\}
\]
and
\[
\overline{{\cal B}}(r) := \left\{ \boldsymbol{\mu}:
\begin{array}{l}
\mu_1 \geq \cdots \geq \mu_{K}, 0 \leq \mu_i \in \R\\
\sum_{i=1}^{K} \mu_i \leq \frac{n_r}{n_r-K+1}(1-r)^+
\end{array} \right\}.
\]
The inequality in \eqref{eq:applargenr2} is due to the fact that ${\cal B}(r) \subseteq \overline{{\cal B}}(r)$, while note that equality holds when $\mu_1=\frac{n_r}{n_r-K+1} (1-r), \ \mu_2=\cdots = \mu_K=0$, which in turn implies that
\[
\bar{c}_\mac(r) =  \min \left\{ r, \frac{K-1}{n_r-K+1}(1-r)^+\right\}
\]
which establishes the corollary.\qed


\section{Description of User-Selection scheme\label{app:userSelectionScheme}}
The proposed user selection algorithm can be drawn from \cite{JiangVar09}, by restricting selection to only transmit antennas (no selection of received antennas). In the setting with $K$ single-antenna users, $\bfH_{\tm{\!eq}}$ (cf. \eqref{eq:H}) has as $k$th column the vector $\bfh_k$ corresponding to the channel vector of the $k$th user.  The goal is to select $L$ out of $K$ users for transmission ($L \leq \nu=\min\{K,n_r\}$). The selection process is closely related to the QR-decomposition of $\bfH_{\tm{\!eq}}$ using a Householder transformation, which exhibits good numerical stability. The selection takes $L$ iterations, and each iteration consists of two steps, a matrix-column permutation followed by a Householder transformation. At the first iteration, the algorithm begins by finding the column $\bfh_{j_1}$ of $\bfH_{\tm{\!eq}}$ with the largest column norm. Then it right-multiplies $\bfH_{\tm{\!eq}}$ by a permutation matrix $\bfPi_1$ to swap $\bfh_1$ and $\bfh_{j_1}$. The second step is to apply a unitary Householder transformation $\bfQ_1$ to $\bfH_{\tm{\!eq}} \bfPi_1$ such that the top-left entry of $\bfQ_1 \bfH_{\tm{\!eq}}\bfPi_1$ is the only nonzero (and positive) entry in the first column. After finishing with the first iteration, the algorithm shifts its focus to the trailing $(n_r-1) \times (K-1)$ submatrix $\bfH_1$ of $\bfQ_1 \bfH_{\tm{\!eq}} \bfPi_1$. Similar to the first iteration, the algorithm identifies the column with the largest norm in $\bfH_1$, swaps it with the first column, and then applies the Householder transformation. Thus, at the end of this iteration, the channel matrix becomes $\bfQ_2 \bfQ_1 \bfH_{\tm{\!eq}} \bfPi_1 \bfPi_2$. The third iteration is the same as the previous ones but focuses on the trailing $(n_r-2) \times (K-2)$ submatrix. The same process is repeated $L$ times, resulting in an output matrix $\bfR_{\tm{JV}}=\bfQ_L \cdots \bfQ_1 \bfH_{\tm{\!eq}} \bfPi_1 \cdots \bfPi_L$. The users associated to the first  $L$ columns of $\bfR_{\tm{JV}}$ are the selected ones~\footnote{It is worth mentioning that if $K=n_r=L$, then $\bfR_{\tm{JV}}$ is an upper triangular matrix, and thus $\bfH_{\tm{\!eq}} \bfPi_1  \cdots \bfPi_L$ has a QR-decomposition equal to $(\bfQ_L \cdots \bfQ_1)^\dag \bfR_{\tm{JV}}$. Also, as the selection focuses only on a trailing submatrix of $\bfQ_{m-1}\cdots \bfQ_1 \bfH_{\tm{\!eq}} \bfPi_1 \cdots \bfQ_{m-1}$ during the $m$-th iteration, one cannot say that the $i$th column of $\bfR_{\tm{JV}}$ has the $i$th largest norm among all columns, for $i=2,\cdots,L$.}.

\section{Proof of Theorem \ref{thm:dus}} \label{proof:dus}

Recalling from the description in Appendix~\ref{app:userSelectionScheme} of the selection algorithm, to select the $L$ users, there are $L$ iterations of applying column permutation matrix $\mathbf{\Pi}_i$ and Householder transformation $\bfQ_i$ to obtain an output matrix $\bfR_{\tm{JV}}$, which is of the form
\begin{IEEEeqnarray*}{rCl}
\bfR_{\tm{JV}} &=& \bfQ_L \cdots \bfQ_1 \bfH_{\tm{\!eq}} \mathbf{\Pi}_1 \cdots \mathbf{\Pi}_L\\
&=&
\left[\begin{array}{ccccccc}
r_{1,1} & * &  \cdots & * &* & \cdots  & * \\
 & r_{2,2}  & \cdots & * & * &\cdots & *\\
 &   & \ddots & \vdots & \vdots & \vdots & \vdots\\
 &   &  & r_{L,L} & * &\cdots & * \\
 &   &  & & \vdots & \vdots & \vdots \\
 &   &  & & * & \cdots & *
\end{array}\right].
\end{IEEEeqnarray*}
Let $u_i$ be the user associated with the $i$th column of $\bfR_{\tm{JV}}$, and let $\{u_1, u_2, \ldots, u_L\}$ be the set of selected users. Since entries of $\bfH_{\tm{\!eq}}$ are i.i.d. $\CN{0}{1}$, to meet an average multiplexing gain $r$ for each user, the selected user has to transmit at a possibly larger multiplexing gain of $\frac{K}{L} r$. For the set of `modified' outage events
\begin{equation}
{\cal O}_{k,\ell} : =  \left\{ \tm{the sum-rate of users $u_{k+1}, \cdots, u_{k+\ell}$ is in outage} \right\}
\end{equation}
for all $k \geq 0$, $\ell \geq 1$, and $k+\ell \leq L$, we will see later on --- in the process of the proof --- that the outage event considered by Jiang and Varanasi~\cite[Theorem 4.1]{JiangVar09}, is a special case of the above events corresponding to $k=0$ and $\ell=L$, i.e., a special case of the outage event ${\cal O}_{0,L}$.

\subsection{DMT Analysis for error event  ${\cal O}_{k,\ell}$}

To analyze the DMT for the error event ${\cal O}_{k,\ell}$ for any $k\geq 0$, $\ell \geq 1$, and $k+\ell \leq L$, let $\bfR_{k,\ell}$ be the matrix resulting from applying $(k+\ell)$ iterations of the Jiang-Varanasi algorithm to overall matrix $\bfH_{\tm{\!eq}}$. We partition matrix $\bfR_{k,\ell}$ as follows:
\bea
\bfR_{k,\ell} &=& \bfQ_{k+\ell} \cdots \bfQ_1 \bfH_{\tm{\!eq}} \mathbf{\Pi}_1 \cdots \mathbf{\Pi}_{k+\ell}\nonumber\\
&=& \left[
\begin{array}{ccc}
\bfR_{L} & \bfR_{C,U} & \bfR_{R,U}\\
		 & \bfR_{C,B} & \bfR_{R,M}\\
& & \bfR_{R,B}
\end{array} \right]
\eea
where
\[
\bfR_{L}  =\left[\begin{array}{ccc}
r_{1,1} & \cdots & r_{1,k} \\
& \ddots & \vdots \\
& & r_{k,k}
\end{array} \right]
\]
\[
\bfR_{C,U} = \left[\begin{array}{ccc}
r_{1,k+1} & \cdots & r_{1,k+\ell} \\
\vdots & \vdots & \vdots \\
r_{k,k+1}& & r_{k,k+\ell}
\end{array} \right]
\]
\[
\bfR_{C,B} = \left[
\begin{array}{ccc}
r_{k+1,k+1} & \cdots & r_{k+1,k+\ell}\\
& \ddots & \vdots \\
& & r_{k+\ell,k+\ell}
\end{array} \right]
\]
\[
\bfR_{R,U} = \left[
\begin{array}{ccc}
r_{1,k+\ell+1} & \cdots & r_{1,K}\\
\vdots & \ddots & \vdots \\
r_{k,k+\ell+1} & \cdots & r_{k,K}
\end{array} \right]
\]
\[
\bfR_{R,M} = \left[
\begin{array}{ccc}
r_{k+1,k+\ell+1} & \cdots & r_{k+1,K}\\
\vdots & \ddots & \vdots \\
r_{k+\ell,k+\ell+1} & \cdots & r_{k+\ell,K}
\end{array}\right]
\]
\[
\bfR_{R,B} = \left[
\begin{array}{ccc}
r_{k+\ell+1,k+\ell+1} & \cdots & r_{k+\ell+1,K}\\
\vdots & \ddots & \vdots \\
r_{n_r,k+\ell+1} & \cdots & r_{n_r,K}
\end{array}\right]
\]
and where the entries satisfy
\begin{equation}
\abs{r_{i,i}}^2 \ \geq \ \sum_{m=i}^{n_r} \abs{r_{m,j}}^2, \label{eq:order1}
\end{equation}
for $i=1, 2, \ldots, k+\ell$ and for $j=i+1, \ldots, K$.

For
\[
\bfR_{C} := \left[
\begin{array}{c}
\bfR_{C,U}\\
\bfR_{C,B}
\end{array} \right]
\]
the probability that the sum-rate of users $u_{k+1}, \cdots, u_{k+\ell}$ results in an outage, takes the form
\begin{IEEEeqnarray}{rCl}
\lefteqn{\Pr\left\{ {\cal O}_{k,\ell}\right\}}\nonumber\\
&=& \Pr \left\{ \log \det \left( \boldsymbol{I}_{n_r} + \snr \bfR_C \bfR_C^\dag \right) \ < \ \frac{K}{L} \ell r \log \snr \right\} \nonumber\\
&=& \int_{{\cal I}(r)} c \cdot \left[ \prod_{i=1}^{k+\ell} p_{\chi_{2(n_r-i+1)}^2}\left(\abs{r_{i,i}}^2 \right) \left( \prod_{j=i+1}^K \frac{1}{\pi} e^{-\abs{r_{i,j}}^2} \right)\right]\nonumber\\
&& \hspace{0.8in} \times  \frac{1}{\pi^{(K-k-\ell)(nr-k-\ell)}} e^{-\norm{\bfR_{R,B}}^2}
\, d \bfR_{k,\ell} \label{eq:int}\\
& : \doteq & \snr^{-d_{k,\ell}\left( \frac{K}{L}\ell r\right)}
\end{IEEEeqnarray}
where $c$ is a constant relating to ordered statistics (cf. the first constraint in \eqref{eq:Dr}), where $p_{\chi_\kappa^2}(\cdot)$ is the probability density function for a $\chi^2$ random variable with $\kappa$ degrees of freedom and with mean $\frac{\kappa}{2}$, and where the integration region is
\begin{equation}
{\cal I}(r)  :=  \left\{ \bfR_{k,\ell}:
\begin{array}{l}
\abs{r_{i,i}}^2 \ \geq \ \sum_{m=i}^{n_r} \abs{r_{m,j}}^2, \\
\quad i=1, \ldots, k+\ell, \quad j=i+1, \ldots, K\\
\det \left( \boldsymbol{I}_{n_r} + \snr \bfR_C \bfR_C ^\dag \right) <  \snr^{\frac{K}{L} \ell r}
 \end{array} \right\}. \label{eq:Dr}
\end{equation}
We have the following three remarks which we will jointly use later on.

\begin{note}\label{remark:app_e_1}
We first note that the diversity exponent associated to $\Pr\left\{ {\cal O}_{k,\ell}\right\}$ from~\eqref{eq:int}, takes the form
\[
d_{k,\ell}\left( \frac{K}{L}\ell r\right)=\lim_{\snr \to \infty} - \frac{\log \Pr\left\{ {\cal O}_{k,\ell}\right\}}{\log \snr} = \inf_{{\bf R}_{k,\ell} \in {\cal I}(r)} D \left( \bfR_{k,\ell} \right)
\]
where the last equality is a manifestation of Laplace's principle \cite{ZheTse}, and where $D \left( \bfR_{k,\ell} \right)$ is a specific diversity-exponent function corresponding to the integrand in~\eqref{eq:int}. To understand the above, we shall look for sub-events of ${\cal I}(r)$ that yield the smallest possible diversity exponent $D \left( \bfR_{k,\ell} \right)$. To this end, note that the entries in $\bfR_L$ and $\bfR_{R,U}$ are not involved in the second constraint in ${\cal I}(r)$, i.e., are not involved in the constraint $\det \left( \boldsymbol{I}_{n_r} + \snr \bfR_C \bfR_C ^\dag \right) <  \snr^{\frac{K}{L} \ell r}$. Furthermore one can see that smaller absolute values of nonzero entries in $\bfR_L$ and $\bfR_{R,U}$ correspond to larger values of $D \left( \bfR_{k,\ell} \right)$. We hence conclude that the dominant sub-event of ${\cal I}(r)$, and the one yielding the smallest diversity exponent $D \left( \bfR_{k,\ell} \right)$), must consist of matrices $\bfR_{k,\ell}$ for which the nonzero entries of the associated submatrices $\bfR_L$ and $\bfR_{R,U}$ are the most typical ones. In other words, the smallest value of $D \left( \bfR_{k,\ell} \right)$ must result from the case when the nonzero entries of submatrices $\bfR_L$ and $\bfR_{R,U}$ have magnitudes in the order of $\snr^0$.
\end{note}

\begin{note}
Let $\lambda_1 \geq \cdots \geq \lambda_\ell$ be the ordered singular values of $\bfR_C$, and let $\nu_1 \geq \cdots \geq \nu_\ell$ be the ordered singular values of $\bfR_{C,B}$. Clearly, as $\bfR_{C,B} \bfR_{C,B}^\dag \preceq \bfR_C \bfR_C^\dag$, we have
\begin{equation}
\nu_i^2 \leq \lambda_i^2, \quad i=1,2,\ldots,\ell. \label{eq:mu_lambda}
\end{equation}
Using~\cite[Lemma 3.3]{JiangVar09}, we have
\begin{equation}
\left(\bfR_{C,B}\right)_{i,i} \ =\  r_{k+i,k+i}^2 \ \geq \ \frac{\sum_{j=i}^\ell \nu_j^2}{\ell-i+1} \ \doteq \ \nu_i^2  \label{eq:lem3.3}
\end{equation}
for $i=1,2,\ldots,\ell$.
Moreover, from \cite[Eq.(27)]{JiangVar09} we have that the squared diagonal elements in $\bfR_{C,B}$ are multiplicatively majorized by its squared singular values, i.e., that
\begin{equation}
\prod_{i=1}^m \nu_i^2 \ \geq \ \prod_{i=1}^m r_{k+i,k+i}^2 \ \dot\geq \ \prod_{i=1}^m \nu_i^2,\quad m=1,2,\ldots,\ell
\end{equation}
where the second dotted inequality is due to~\eqref{eq:lem3.3}. It then follows that
\begin{equation}
r_{k+i,k+i}^2 \doteq \nu_i^2 \label{eq:rmu}
\end{equation}
for $i=1,2,\ldots,\ell$.
\end{note}

\begin{note}
After the first $k$ iterations of the Jiang-Varanasi algorithm, we get
\begin{multline*}
\bfQ_{k} \cdots \bfQ_1 \bfH_{\tm{eq}} \mathbf{\Pi}_1 \cdots \mathbf{\Pi}_k \\
= \left[
\begin{array}{ccccccc}
r_{1,1} & * &  \cdots & * &* & \cdots  & * \\
 & r_{2,2}  & \cdots & * & * &\cdots & *\\
 &   & \ddots & \vdots & \vdots & \vdots & \vdots\\
 &   &  & r_{k,k} & * &\cdots & * \\
 &   &  & & \vdots & \vdots & \vdots \\
 &   &  & & * & \cdots & *
\end{array}\right]\\ \ = \ \left[
\begin{array}{cc}
\begin{array}{c}
\bfR_{L} \\
\boldsymbol{0}
\end{array}
& \bfR_R
\end{array}\right]
\end{multline*}
where $\bfR_R$ is the $(n_r \times (K-k))$ matrix consisting of the rightmost $(K-k)$ columns of the above matrix. In Remark \ref{remark:app_e_1} we argued that the dominant matrices $\bfR_{k,\ell}$ in event ${\cal I}(r)$ --- yielding the smallest possible diversity exponent $D\left( \bfR_{k,\ell}\right)$ --- must have submatrices $\bfR_L$ whose nonzero entries are of magnitude that is in the order of $\snr^0$. This implies that the entries of $\bfR_R$ can still be regarded as i.i.d. $\CN{0}{1}$ random variables since any ensemble of random $\CN{0}{1}$ variables has the same asymptotic probability as those subject to an additional constraint of magnitude lesser than $\snr^0$. Consequently, the ordered singular values $\lambda_i, i=1, \ldots,\ell$, of matrix $\bfR_C$ are still of the same joint probability density function as that for the ordered singular values of an $(n_r \times \ell)$ random matrix with i.i.d. $\CN{0}{1}$ entries.
\end{note}

With the above three remarks in place, we now proceed to analyze the integral~\eqref{eq:int} to obtain a formula for the DMT function $d_{k,\ell}(r)$. Specifically, we will show that
\begin{multline}
d_{k,\ell}(r) = \inf_{{\cal A}_\ell(r)} \biggl\{ \sum_{i=1}^\ell  (n_r + \ell - 2i+1)  \alpha_i +
\sum_{i=1}^{\ell-1} (K-k-\ell)\alpha_i\\
+\alpha_\ell (K-k-\ell)(n_r-k-\ell+1)
\biggr\} \label{eq:duk}
\end{multline}
where
\[
{\cal A}_\ell(r) = \left\{ 0 \leq \alpha_1 \leq \cdots \leq \alpha_\ell : \sum_{i=1}^\ell (1-\alpha_i)^+ \leq r \right\}.
\]
To see the above, set $\lambda_i^2 \doteq \snr^{-\alpha_i}$ with $\alpha_1 \leq \alpha_2 \cdots \leq \alpha_\ell$. The first summand appearing in \eqref{eq:duk} follows from the aforementioned fact that the joint probability density function of ordered singular values $\lambda_1 \geq \cdots \geq \lambda_\ell$ for an $(n_r \times \ell)$ matrix with i.i.d. $\CN{0}{1}$ entries is (cf.~\cite{ZheTse, JiangVar09})
\[
p(\alpha_1, \cdots, \alpha_\ell) \ \doteq \ \snr^{- \sum_{i=1}^\ell  (n_r + \ell - 2i+1)  \alpha_i }.
\]
Also by \eqref{eq:order1}, \eqref{eq:mu_lambda}, and \eqref{eq:rmu}, we have the following constraints for the entries in matrix $\bfR_{k,\ell}$:
\begin{enumerate}
\item For $i=1, \ldots, \ell-1$, the entries $r_{k+i,j}$ ($j=k+\ell+1, \ldots, K$) must satisfy
\[
\abs{r_{k+i,j}}^2 \leq \abs{r_{k+i,k+i}}^2 \leq \lambda_i^2.
\]
The constraints on $r_{k+i,j}$ contribute to \eqref{eq:duk} the second summation, i.e. the term $\sum_{i=1}^{\ell-1} (K-k-\ell)\alpha_i$. \medskip

\item Entries $r_{k+i,j}$ with $i=\ell, \ldots, n_r-k$ and $j=k+\ell+1, \ldots, K$, must satisfy
\[
\sum_{i=\ell}^{n_r-k} \abs{r_{k+i,j}}^2 \leq \abs{r_{k+\ell,k+\ell}}^2 \doteq \mu_\ell^2 \leq \lambda_\ell^2
\]
Such constraints contribute to \eqref{eq:duk} the last summand, i.e., the term $\alpha_\ell (K-k-\ell)(n_r-k-\ell+1) $.
\end{enumerate}
This completes the proof of~\eqref{eq:duk}.
$\\$
Finally, the proof of Theorem~\ref{thm:dus} is complete after noting that the union of outage events ${\cal O}_{k,\ell}$, is a subset of the overall outage event, i.e., that
\[
\bigcup_{k \geq 0, \ell \geq 1, \atop k+\ell \leq L} {\cal O}_{k,\ell}
\subseteq \bigcup_{{\cal U} \subset \{u_1, \ldots, u_L\}} \left\{ \boldsymbol{H}_{\tm{eq}} : \tm{ users in $\cal U$ are in outage}\right\}.
\]
\qed

\begin{note}
It is interesting to see that the antenna-selection DMT in Jiang and Varanasi (\cite[Theorem 4.1]{JiangVar09}), can be derived as a special case of~\eqref{eq:duk}. Specifically it can be shown that when $k=0$ and $\ell=L \leq \nu$, the corresponding $d_{0,L} (r)$ coincides with the DMT in \cite{JiangVar09}, as both functions are piecewise linear, connecting the following $(P+2)$ points
\begin{equation}
(r,(K-r)(n_r-r)), r=0,1,\ldots,P, \tm{ and } (L,0),
\end{equation}
where
\begin{equation}
P \ = \ \arg \min_{p=0,1,\ldots,L-1} \frac{(K-p)(n_r-p)}{L-p}.
\end{equation}
\end{note}

\section*{Acknowledgments}
The research leading to these results has received funding under the European Community's Seventh Framework Programme (FP7/2007-2013) grant agreement 318306 (NEWCOM\#), from the FP7 CELTIC SPECTRA project, from Agence Nationale de la Recherche project ANR-IMAGENET, and from the Taiwan Ministry of Science and Technology under grants MOST-101-2923-E-009-001-MY3 and MOST 103-2221-E-009 -043 -MY3.


\begin{thebibliography}{100}
\bibitem{JaldenE12}
J.~Jald\'{e}n and P.~Elia, ``Sphere decoding complexity exponent for decoding
  full-rate codes over the quasi-static {MIMO} channel,'' \emph{{IEEE} Trans.
  Inf. Theory}, vol.~58, no.~9, pp. 5785--5803, Sep. 2012.

\bibitem{SinghEJ12}
A. Singh, P.~Elia, and J.~Jald\'{e}n, ``Achieving a vanishing {SNR} gap to
  exact lattice decoding at a subexponential complexity,'' \emph{{IEEE} Trans.
  Inf. Theory}, vol.~58, no.~6, pp. 3692--3707, Jun. 2012.

\bibitem{conf_SiElJa12}
A.~Singh, P.~Elia, , and J.~Jald\'{e}n, ``Complexity analysis for {ML}-based
  sphere decoder achieving a vanishing performance-gap to brute force {ML}
  decoding,'' in \emph{Proc. Int. Zurich Seminar on Communications (IZS)}, Mar.
  2012, pp. 127--130.

\bibitem{conf_EliJal10}
P.~Elia and J.~Jald\'{e}n, ``General {DMT} optimality of {LR}-aided linear
  {MIMO-MAC} transceivers with worst-case complexity at most linear in
  sum-rate,'' in \emph{Proc. {IEEE} Information Theory Workshop {(ITW)}},
  Jan. 2010.


\bibitem{TseVisZhe}
D.~N.~C. Tse, P.~Viswanath, and L.~Zheng, ``Diversity-multiplexing tradeoff in
  multiple-access channels,'' \emph{{IEEE} Trans. Inf. Theory}, vol.~50, no.~9,
  pp. 1859--1874, Sept. 2004.


\bibitem{ZheTse}
L.~Zheng and D.~N.~C. Tse, ``Diversity and multiplexing: a fundamental tradeoff
  in multiple antenna channels,'' \emph{{IEEE} Trans. Inf. Theory}, vol.~49,
  no.~5, pp. 1073--1096, May 2003.


\bibitem{JitRaj13} G.R. Jithamithra and B. Sundar Rajan, ''{Minimizing the complexity of fast sphere decoding of STBCs},'' \emph{{IEEE} Trans. Wireless Commun.}, vol.~12, no.~12, pp. 6142--6153, Dec. 2013.

\bibitem{NatSriRaj13} L. P. Natarajan, K. Pavan Srinath and B. Sundar Rajan, ''On the sphere decoding complexity of high rate multigroup decodable STBCs in asymmetric MIMO systems,'' \emph{{IEEE} Trans. Inf. Theory}, vol.~59, no.~9, pp. 5959--5965, Sep. 2013.

\bibitem{JalBarOtt09} J. Jald\'en and L.G. Barbero, B. Ottersten, and J.S. Thompson, ``The error probability of the fixed-complexity sphere decoder,''
\emph{{IEEE} Trans. Signal Processing}, vol.~57, no.~7, pp. 2711--2720, Jul. 2009.

\bibitem{HowSirCal} S. Sirinaunpiboon, Calderbank, A.R.; Howard, S.D. ``Fast essentially maximum likelihood decoding of the Golden Code'',  \emph{{IEEE} Trans. Inf. Theory}, vol.~57, no.~6, pp. 3537--3541, Jun. 2011.

\bibitem{NatRaj13b} L.P. Natarajan and B.S. Rajan, ``An adaptive conditional zero-forcing decoder with full-diversity, least complexity and essentially-ML performance for STBCs,'' \emph{{IEEE} Trans. Inf. Theory}, vol.~61, no. 2, pp. 253--263, Jan.~2013.

\bibitem{SEJ:13b} A. Singh, P. Elia, and J.~Jald\'{e}n, ``Rate-reliability-complexity tradeoff for ML and lattice decoding of full-rate codes,'' in
  \emph{Proc. 2013 IEEE Int. Symp. Inform. Theory}, ISIT 2013.

\bibitem{JaldenE10}
J.~Jald\'{e}n and P.~Elia, ``DMT optimality of {LR}-aided linear decoders for a general class of channels, lattice designs, and system models,''
\emph{{IEEE} Trans. Inf. Theory}, vol.~56, no.~10, pp. 4765--4780, Sep. 2010.




\bibitem{JaldenEisit09}
J.~Jald\'{e}n and P.~Elia, ``{LR}-aided {MMSE} lattice decoding is {DMT} optimal for all approximately universal codes,'' in
  \emph{Proc. 2009 IEEE Int. Symp. Inform. Theory}, ISIT 2009.

\bibitem{Kuser}
H.~F. Lu, C.~Hollanti, R.~Vehkalahti, and J.~Lahtonen, ``{DMT} optimal codes
  constructions for multiple-access {MIMO} channel,'' \emph{{IEEE} Trans. Inf.
  Theory}, vol.~57, no.~6, pp. 3594--3617, June 2011.


\bibitem{remarkdmt}
H.~F. Lu, ``Remarks on diversity-multiplexing tradeoffs for multiple-access and
  point-to-point {MIMO} channels,'' \emph{{IEEE} Trans. Inf. Theory}, vol.~58,
  no.~2, pp. 858--863, Feb. 2012.

\bibitem{LVHLHV}
H.-F. Lu, R.~Vehkalahti, C.~Hollanti, J.~Lahtonen, Y.~Hong, and E.~Viterbo,
  ``New space-time code construcions for two user multiple access channels,''
  \emph{{IEEE} J. Sel. Topics Signal Process.}, vol.~3, no.~6, pp. 939--957,
  Dec. 2009.
  
\bibitem{KEJ:09} S. Kittipiyakul, P. Elia, and T. Javidi, ``High­-SNR analysis of outage­-limited communications with bursty and delay -limited information,''  \emph{{IEEE} Trans. Inf. Theory}, vol.~55, no.~2, pp. 746--763, Feb. 2009.
  

\bibitem{DCV:13} T. Datta, A. Chockalingam, and E. Viterbo, ``Gaussian sampling based lattice decoding,'' in
  \emph{Proc. 2013 IEEE Int. Symp. Inform. Theory}, ISIT 2013.
  
\bibitem{SDC:13} K. Singhal, T. Datta, and A. Chockalingam, ``Lattice reduction aided detection in large-MIMO systems,''   \emph{Proc. 2013 IEEE Int. workshop on Signal Processing Advances in Wireless Communications}, SPAWC 2013.
    
\bibitem{PMJ:14}
J. Pan, Wing-Kin Ma, and J.~Jald\'{e}n, ``MIMO detection by Lagrangian dual maximum-likelihood relaxation: Reinterpreting regularized lattice decoding,'' \emph{{IEEE} Trans. Signal Processing}, vol. 62, no. 2, pp. 511--524, Jan. 2014.

\bibitem{EOK:07}
P.~Elia, F.~Oggier, and P. V. Kumar, ``Asymptotically optimal cooperative wireless networks with reduced signaling complexity,'' \emph{{IEEE} Journal on Selected Areas in Communications}, vol.~25, no.~2, pp. 258--267, Feb. 2007.

\bibitem{EK:05}
P.~Elia and P. V. Kumar, ``Approximately universal optimality over several dynamic and non-dynamic cooperative diversity schemes for wireless networks.'' Dec. 7, 2005 [Online]. Available: http://arxiv.org/pdf/cs.IT/0512028

\bibitem{JO:05}
J.~Jald\'{e}n and B. Ottersten, ``On the complexity of sphere decoding in digital communications,'' \emph{{IEEE} Trans. Signal Processing}, vol. 53, no. 4, pp. 1474–-1484, Apr. 2005.


\bibitem{JiangVar09}
Y.~Jiang and M.~K. Varanasi, ``The {RF}-chain limited {MIMO} system- part {I}:
  optimum diversity-multiplexing tradeoff,'' \emph{{IEEE} Trans. Wireless
  Commun.}, vol.~8, no.~10, pp. 5238--5247, Oct. 2009.

\bibitem{conf_DaElCa04}
M.~O. Damen, H.~{El Gamal}, and G.~Caire, ``{MMSE-GDFE} lattice decoding for
  solving under-determined linear systems with integer unknowns,'' in
  \emph{Proc. 2004 IEEE Int. Symp. Inform. Theory}, ISIT 2004.
  
  \bibitem{TseVis}
D.~Tse and P.~Viswanath, \emph{Fundamentals of Wireless Communication}.\hskip
  1em plus 0.5em minus 0.4em\relax Cambridge, UK: Cambridge, University Press,
  2005.


\end{thebibliography}
\end{document}